\documentclass[12pt]{article}
\pdfoutput=1

\usepackage{mathtools} 
\usepackage{amstext} 
\usepackage{amssymb} 
\usepackage{bbold} 
\usepackage{amsthm} 
\usepackage{stmaryrd} 

\usepackage[utf8]{inputenc}

\usepackage{relsize} 
\usepackage{microtype} 
\usepackage{fullpage} 
\usepackage{csquotes} 
\usepackage{ragged2e} 
\usepackage{cancel}
\usepackage{pifont}

\usepackage{hyperref} 
\usepackage[compress]{cite} 

\usepackage{graphicx} 
\usepackage[usenames,dvipsnames]{xcolor} 
\usepackage{stackengine}
\usepackage{tikz} 
\usepackage{circuitikz} 
\usetikzlibrary{
    arrows,
    shapes,
    decorations,
    intersections,
    backgrounds,
    positioning,
    circuits.ee.IEC
    }
\usepackage{tikz-cd}

\usepackage{rotating}
\usepackage{floatpag}
\usepackage{xifthen}
\usepackage[normalem]{ulem}

\usepackage{dsfont}

\newcounter{theorem_c}

\theoremstyle{plain}

\newtheorem{theorem}[theorem_c]{Theorem}

\newtheorem{proposition}[theorem_c]{Proposition}

\newtheorem*{theorem*}{Theorem}
\newtheorem*{lemma*}{Lemma}
\newtheorem*{corollary*}{Corollary}
\newtheorem*{proposition*}{Proposition}

\newtheorem{definition}[theorem_c]{Definition}

\newtheorem{remark}[theorem_c]{Remark}

\newtheorem*{definition*}{Definition}
\newtheorem*{example*}{Example}
\newtheorem*{remark*}{Remark}

\newcommand{\suchthat}[2]{\left\{\,#1\,\middle|\,#2\,\right\}}
\newcommand{\Forall}[2]{\forall\,#1.\;#2}

\newcommand{\powerset}[1]{\mathcal{P}(#1)}

\newcommand{\integers}{\mathbb{Z}}
\newcommand{\modclass}[2]{#1 \; (\text{mod } #2)}

\newcommand{\reals}{\mathbb{R}}

\newcommand{\nonstd}[1]{\,^\star #1}
\newcommand{\starIntegers}{\nonstd{\integers}} 

\newcommand{\id}[1]{{id}_{#1}}

\newcommand{\obj}[1]{\operatorname{obj}\left(#1\right)}

\newcommand{\ket}[1]{\vert #1 \rangle}

\newcommand{\Hom}[3]{\operatorname{Hom}_{\,#1}\left[#2,#3\right]}

\newcommand{\Automs}[2]{\operatorname{Aut}_{\,#1}\left[#2\right]}

\newcommand{\setCat}{\operatorname{Set}}

\newcommand{\posetCat}{\operatorname{Pos}}
\newcommand{\fhilbCat}{\operatorname{fHilb}}
\newcommand{\hilbCat}{\operatorname{Hilb}}
\newcommand{\cpmCat}[1]{\operatorname{CPM}\left[#1\right]}
\newcommand{\cpstarCat}[1]{\operatorname{CP*}\left[#1\right]}
\newcommand{\splitCat}[1]{\operatorname{Split}\left[#1\right]}

\newcommand{\starHilbCat}{^\star\!\hilbCat}

\tikzstyle{ground_r}=[circuit ee IEC,thick,ground,rotate=90]
\newcommand{\traceSym}{
    \begin{tikzpicture}
        \path[use as bounding box] (-0.1, -0.1) rectangle (0.15, 0.1);
        \node (0) at (0, -0.1) {};
        \node (1) at (0, 0.1) [ground_r, scale=0.5] {};
        \draw[-] (0.center) to (1);
    \end{tikzpicture}
}
\newcommand{\trace}[1]{\traceSym_{#1}}

\newcommand{\future}[1]{J^{+}\left(#1\right)}
\newcommand{\past}[1]{J^{-}\left(#1\right)}
\newcommand{\futuredom}[1]{D^{+}\left(#1\right)}
\newcommand{\pastdom}[1]{D^{-}\left(#1\right)}
\newcommand{\sliceCat}[1]{\operatorname{Slices}\left(#1\right)}
\newcommand{\cauchySliceCat}[1]{\operatorname{CauchySlices}\left(#1\right)}
\newcommand{\slicePreceq}{\twoheadrightarrow}
\newcommand{\causOrdCat}{\operatorname{CausOrd}}
\newcommand{\boundedRegionCat}[1]{\operatorname{Regions}_{\text{bnd}}\left(#1\right)}
\newcommand{\regionCat}[1]{\operatorname{Regions}\left(#1\right)}
\newcommand{\states}[2]{\operatorname{States}_{#1}\left(#2\right)}
\newcommand{\statesPresheaf}[1]{\operatorname{States}_{#1}}

\newcommand{\fcstaralgCat}{\operatorname{fC*}\!\operatorname{alg}}
\newcommand{\cstaralgCat}{\operatorname{C*}\!\operatorname{alg}}
\newcommand{\wstaralgCat}{\operatorname{W*}\!\operatorname{alg}}
\newcommand{\vnalgCat}{\operatorname{vNA}}

\begin{document}

\title{Functorial evolution of quantum fields}
\author{
    Stefano Gogioso,
    Maria E. Stasinou
    and Bob Coecke\\
    {\small Department of Computer Science, University of Oxford, Oxford, UK}\\
    {\small \texttt{\{stefano.gogioso, maria.stasinou, bob.coecke\}(at)cs.ox.ac.uk}}
}
\date{}

\maketitle

\begin{abstract}
    \noindent
    We present a compositional algebraic framework to describe the evolution of quantum fields in discretised spacetimes.
    We show how familiar notions from Relativity and quantum causality can be recovered in a purely order-theoretic way from the causal order of events in spacetime, with no direct mention of analysis or topology.
    We formulate theory-independent notions of fields over causal orders in a compositional, functorial way.
    We draw a strong connection to Algebraic Quantum Field Theory (AQFT), using a sheaf-theoretical approach in our definition of spaces of states over regions of spacetime.
    We introduce notions of symmetry and cellular automata, which we show to subsume existing definitions of Quantum Cellular Automata (QCA) from previous literature.
    Given the extreme flexibility of our constructions, we propose that our framework be used as the starting point for new developments in AQFT, QCA and more generally Quantum Field Theory.
\end{abstract}

\section{Introduction}

Like much of classical physics, the study of Relativity and quantum field theory has deep roots in topology and geometry. However, recent years have seen a steady shift from the traditional approaches to more a more abstract algebraic perspective, based on the identification of spacetime structure with \emph{causal order}.

This new way of looking at causality finds its origin in a much-celebrated result by Malament \cite{malament1977class}, itself based on previous work by Kronheimer, Penrose, Hawking, King and McCarthy \cite{kronheimer1967structure,hawking1976new}.
If $M$ is a Lorentzian manifold, we say that $M$ is \emph{future- (resp. past-) distinguishing} iff two events $x, y \in M$ (i.e. two spacetime points) having the same exact causal future (resp. past) are necessarily identical
\footnote{The requirement for a manifold $M$ to be future- and past- distinguishing is essentially one of well-behaviour, e.g. excluding causal violations such as closed timelike curves (all points of which necessarily have the same causal past and future).}.
Given a Lorentzian manifold $M$, we can define a partial order $\leq_M$ between its events---the \emph{causal order}---by setting $x \leq_M y$ iff $x$ \emph{causally precedes} $y$ in $M$, i.e. iff there exists a future-directed causal curve---a smooth curve in $M$ with everywhere future-directed time-like or light-like tangent vector---from $x$ to $y$.
The 1977 result by Malament \cite{malament1977class} can then be stated as follows.

\begin{theorem}\label{theorem:malament1977}
    Let $M$ and $M'$ be two Lorentzian manifolds, both manifolds being future-and-past--distinguishing. The associated causal orders $(M, \leq_M)$ and $(M, \leq_{M'})$ are order-isomorphic if and only if $M$ and $M'$ are conformally equivalent.
\end{theorem}

\noindent While the result by Malament guarantees that future-and-past--distinguishing manifolds (up to conformal equivalence) can be identified with their causal orders, it does not provide a characterisation of which partial orders arise as causal orders on manifolds (or restrictions thereof to manifold subsets). This lack of exact correspondence between topology and order is the motivation behind many past and current lines of enquiry, such as the domain-theoretic work of \cite{martin2010domain,martin2012spacetime}. The complete characterisation of which partial orders arise as the causal orders of Lorentzian manifolds is still an open question.

A different approach to the order-theoretic study of spacetime is given by the \emph{causal sets} research programme (cf. \cite{bombelli1987space,bombelli2006overview}).
A \emph{causal set} is a poset which is \emph{locally finite}, i.e. such that for every $x, y \in C$ the subset $\suchthat{z \in C}{x \leq z \leq y}$ is finite.
\footnote{The local finiteness condition for a causal set can equivalently be stated as the requirement that the partial order arises as the reflexive-transitive closure of a non-transitive directed graph, its Hasse diagram (see e.g. \cite{pinzani2019categorical}).}
Causal sets arise as discrete subsets of Lorentzian manifolds (under the causal order inherited by restriction) and a fundamental pursuit for the community is a characterisation of the large-scale properties of spacetime as emergent from a discrete small-scale structure.
In particular, the question whether a causal set can always be (suitably) embedded as a discrete subset of a Lorentzian manifold is central to the programme and---as far as we are aware---one which is still to be completely answered \cite{bombelli2006overview}.

When it comes to incorporating quantum fields into the spacetimes, efforts have mostly been focused in three significant directions: algebraic approaches, topological approaches and quantum cellular automata.

The algebraic approaches take a functorial and sheaf-theoretic view of quantum fields, studying the \emph{local} structure of fields through the algebras of observables---usually C*-algebras or von Neumann algebras---over the regions of spacetime. Prominent examples include \emph{Algebraic Quantum Field Theory} (AQFT) \cite{haag1964algebraic,halvorson2006algebraic} and the topos-theoretic programmes \cite{heunen2009topos,doring2008thing}.

The topological approaches focus instead on global aspects of relativistic quantum fields, foregoing any possibility of studying local structure by requiring that field theories be \emph{topological}, i.e. invariant under large scale deformations of spacetime. The resulting \emph{Topological Quantum Field Theories} (TQFTs) \cite{lurie2009classification,atiyah1988topological,witten1988topological} have achieved enormous success in fields such as condensed matter theory and quantum error correction.
Like AQFT, TQFTs have a categorical formulation as functors from a category of spacetime ``pieces'' to categories of Vector spaces and algebras. The difference is in the \emph{nature} of those ``pieces'': in AQFT a spacetime is given and the order structure of its regions is considered; in TQFT, on the other hand, (equivalence classes of) basic topological manifolds are given, which can be combined together to form myriad different spacetimes.

The approaches based on Quantum Cellular Automata (QCA) \cite{dariano2016automata,arrighi2019automata,vonNeumann1966automata}, finally, attempt to tame the issues with the formulation of quantum field theory by positing that full-fledged quantum fields in spacetime can be understood as the continuous limit of much-more-manageable theories, dealing with quantum fields living on discrete lattices and subject to discrete time evolution (known as Quantum Cellular Automata).

In this work, we propose to use tools from category theory to unify key aspects of the approaches above under a single generalised framework. Specifically, our work is part of an effort to gain an operational, process-theoretic understanding of the relationship between quantum theory and Relativistic causality \cite{coecke2016terminality,coecke2013causal,kissinger2017categorical,pinzani2019categorical}.
Our key contribution, across the next four sections, will be the formulation of a functorial and theory-independent notion of field theory based solely on the order-theoretic structure of causality.
To exemplify the flexibility of our construction, in Section~\ref{section:connection-aqft} we will build a strong connection to Algebraic Quantum Field Theory, based on a sheaf-theoretic formulation of states over regions.
In Section~\ref{section:connection-qca}, finally, we will formulate a notion of cellular automaton which encompasses and greatly generalises notions of QCA from existing literature.

\section{Causal orders}
\label{section:causal-orders}

In this work, we will consider posets as an abstract model of causally well-behaved spacetimes.
This means that we will be working in the category $\posetCat$ of posets and monotone maps between them, with Malament's result \cite{malament1977class} showing that future-and-past--distinguishing conformal Lorentzian manifolds embed into $\posetCat$.
To highlight the intended relationship to spacetimes, we will refer to partial orders as \emph{causal orders} for the remainder of this work.

\begin{definition}\label{definition:causal-order}
    By a \emph{causal order} we mean a poset $\Omega = (|\Omega|, \leq)$, i.e. a set $|\Omega|$ equipped with a partial order $\leq$ on it.
    We refer to the elements of $\Omega$ as \emph{events}.
    Given two events $x, y \in \Omega$ we say that $x$ \emph{causally precedes} $y$ (equivalently that $y$ \emph{causally follows} $x$) iff $x \leq y$.
    We say that $x$ and $y$ are \emph{causally related} iff $x \leq y$ or $y \leq x$.
    A \emph{causal sub-order} $\Omega'$ of a causal order $\Omega$ is a subset $|\Omega'| \subseteq |\Omega|$ endowed with the structure of a poset by restriction.
    \footnote{
        I.e. such that for all $x, y \in |\Omega'|$ we have that $x \leq y$ in $\Omega'$ if and only if $x \leq y$ in $\Omega$.
    }
\end{definition}

\noindent As we now proceed to demonstrate, several familiar concepts from Relativity can be defined in a purely combinatorial manner on partial orders.

\subsection{Causal Paths}

\begin{definition}\label{definition:causal-paths}
    Let $\Omega$ be a causal order and let $x, y \in \Omega$ be two events. A \emph{causal path} from $x$ to $y$ is a maximal totally ordered subset $\gamma \subseteq \Omega$ such that $x = \min \gamma$ and $y = \max \gamma$. Maximality of the subset $\gamma \subseteq \Omega$ here means that there is no total order $\gamma' \subseteq \Omega$ strictly containing gamma and such that $x = \min \gamma'$ and $y = \max \gamma'$. We write $\gamma: x \rightsquigarrow y$ to denote that $\gamma$ is a causal path from $x$ to $y$.
\end{definition}

\noindent The causal diamond from $x$ to $y$ in a causal order $\Omega$ is the union of all causal paths $x \rightsquigarrow y$ in $\Omega$. Furthermore, causal paths in $\Omega$ can be naturally organised into a category as follows:

\begin{itemize}
    \item the objects are the events $x \in \Omega$;
    \item the morphisms from $x$ to $y$ are the paths $x \rightsquigarrow y$;
    \item the identity morphism on $x$ is the singleton path $\{x\}: x \rightsquigarrow x$;
    \item composition of two paths $\gamma: x \rightsquigarrow y$ and $\xi: y \rightsquigarrow z$ is the set-theoretic union of the subsets $\gamma, \xi \subseteq \Omega$:
    \begin{equation}
        \xi \circ \gamma := (\xi \cup \gamma): x \rightsquigarrow z
    \end{equation}
\end{itemize}

\begin{definition}\label{definition:causal-future-past}
    Let $\Omega$ be a causal order and let $x \in \Omega$ be an event. The \emph{causal future} $\future{x}$ of $x$ is the set of all events $y$ which causally follow it:
    \begin{equation}
        \future{x} := \suchthat{y \in \Omega}{x \leq y}
    \end{equation}
    Similarly, the \emph{causal past} $\past{x}$ of $x$ is the set of all events $y$ which causally precede it:
    \begin{equation}
        \past{x} := \suchthat{y \in \Omega}{y \leq x}
    \end{equation}
    We also define causal future and past for arbitrary subsets $A \subseteq \Omega$ by union:
    \begin{equation}
        \future{A} := \bigcup_{x \in A} \future{x}
        \hspace{2cm}
        \past{A} := \bigcup_{x \in A} \past{x}
    \end{equation}
\end{definition}

\begin{remark}
    A causal order $\Omega$ is automatically future-and-past--distinguishing. To see this, assume that $\future{x} = \future{y}$ for some $x, y \in \Omega$: then both $x \in \future{x} = \future{y}$, implying $y \leq x$, and $y \in \future{y} = \future{x}$, implying $x \leq y$, so that $x = y$ by antisymmetry of the partial order $\leq$. The assumption that $\past{x} = \past{y}$ analogously implies that $x = y$.
\end{remark}

\begin{definition}\label{definition:unbounded-causal-paths}
    Let $\Omega$ be a causal order and let $x \in \Omega$ be an event.
    By a causal path $\gamma: x \rightsquigarrow +\infty$ (resp. $\gamma: -\infty \rightsquigarrow x$) we denote a maximal totally ordered subset $\gamma \subseteq \Omega$ such that $x = \min \gamma$ (resp. $x = \max \gamma$).
    If $\Omega$ has a global maximum (resp. global minimum), then we denote it by $+\infty$ (resp. $-\infty$) for consistency with our previous definition of causal paths, otherwise the symbol $+\infty$ (resp. $-\infty$) is never used to denote an actual element of $C$.
\end{definition}

\noindent The causal future (resp. causal past) of an event $x$ is the union of all causal paths $x \rightsquigarrow +\infty$ (resp. $-\infty \rightsquigarrow x$).

\subsection{Space-like Slices}

\begin{definition}\label{definition:domain-of-dependence}
    Let $\Omega$ be a causal order and let $A \subseteq \Omega$ be any subset.
    The \emph{future domain of dependence} $\futuredom{A}$ of $A$ is the subset of all events $x \in \Omega$ which ``necessarily causally follow $A$'', in the sense that every causal path $-\infty \rightsquigarrow x$ intersects $A$:
    \begin{equation}
        \futuredom{A}
        :=
        \suchthat{x \in \Omega}{
            \Forall{\gamma: -\infty \rightsquigarrow x}{\gamma \cap A \neq \emptyset}
        }
    \end{equation}
    The \emph{past domain of dependence} $\pastdom{A}$ of $A$ is the subset of all events $x \in \Omega$ which ``necessarily causally precede $A$'', in the sense that every causal path $x \rightsquigarrow +\infty$ intersects $A$:
    \begin{equation}
        \pastdom{A}
        :=
        \suchthat{x \in \Omega}{
            \Forall{\gamma: x \rightsquigarrow +\infty}{\gamma \cap A \neq \emptyset}
        }
    \end{equation}
\end{definition}

\noindent The domains of dependence of a subset $A$ are related to its past and future by the following two Propositions.

\begin{proposition}\label{proposition:domain-of-dependence-past-future-1}
    Let $\Omega$ be a causal order and let $A \subseteq \Omega$ be any subset. Then $\futuredom{A} \subseteq \future{A}$ and $\pastdom{A} \subseteq \past{A}$.
\end{proposition}
\begin{proof}
    Let $x \in \futuredom{A}$ be any event in the future domain of dependence of $A$.
    The set of causal paths $-\infty \rightsquigarrow x$ is necessarily non-empty, because there must be at least one such path extending the singleton path $\{x\}: x \rightsquigarrow x$. Let $\gamma: -\infty \rightsquigarrow x$ be one such path.
    Because $x \in \futuredom{A}$, $\gamma$ must intersect $A$ at some point $y \leq x$, and we define $\gamma' := \gamma \cap \future{\{y\}} \neq \emptyset$.
    By definition, $y = \min \gamma'$.
    Because $\future{\{y\}}$ is upward-closed, $x = \max \gamma'$ and $\gamma': y \rightsquigarrow x$ is such that $\gamma' \subseteq \future{\{y\}} \subseteq \future{A}$, so we conclude that $x \in \future{A}$.
    The proof that $\pastdom{A} \subseteq \past{A}$ is analogous.
\hfill$\square$\end{proof}

\begin{proposition}\label{proposition:domain-of-dependence-past-future-2}
    Let $\Omega$ be a causal order and let $A \subseteq \Omega$ be any subset. If $B \subseteq \futuredom{A}$ then $\future{B} \subseteq \future{A}$ and $\past{B} \subseteq \past{A} \cup \future{A}$. Dually, if $B \subseteq \pastdom{A}$ then $\past{B} \subseteq \past{A}$ and $\future{B} \subseteq \past{A} \cup \future{A}$.
\end{proposition}
\begin{proof}
    Without loss of generality, assume $B \subseteq \futuredom{A}$---the case $B \subseteq \pastdom{A}$ is proven analogously.
    From Proposition~\ref{proposition:domain-of-dependence-past-future-1} we have that $B \subseteq \futuredom{A} \subseteq \future{A}$, so we conclude that $\future{B} \subseteq \future{A}$ by upward-closure of $\future{A}$.
    Now consider $x \in \past{B}$. Let $\gamma: x \rightsquigarrow y$ be any path with $y \in B$ and let $\gamma': -\infty \rightsquigarrow y$ be any path extending $\gamma$.
    Because $B \subseteq \futuredom{A}$, the intersection $\gamma' \cap A$ contains at least some point $z$.
    Because $\gamma'$ is totally ordered, we have two possible cases: $z \leq x$ and $z \geq x$.
    If $z \leq x$, then $\gamma' \cap \future{\{z\}} \cap \past{\{x\}}: z \rightsquigarrow x$ shows that $x \in \future{A}$.
    If $z \geq x$, then $\gamma' \cap \past{\{z\}} \cap \future{\{x\}}: x \rightsquigarrow z$ shows that $x \in \past{A}$.
\hfill$\square$\end{proof}

\begin{definition}\label{definition:spacelike-slices}
    Let $\Omega$ be a causal order.
    We say that two events $x, y$ are \emph{space-like separated} if they are not causally related, i.e. if neither $x \leq y$ nor $y \leq x$.
    Consequently, we define a \emph{(space-like) slice} $\Sigma$ in $\Omega$ to be an antichain, i.e. a subset $\Sigma \subseteq \Omega$ that $(\Sigma, \leq)$ is a discrete partial order (equivalently,  any two distinct $x, y \in \Sigma$ are space-like separated).
\end{definition}

\begin{figure}[h]
    \begin{center}
        \includegraphics[width=0.45\textwidth]{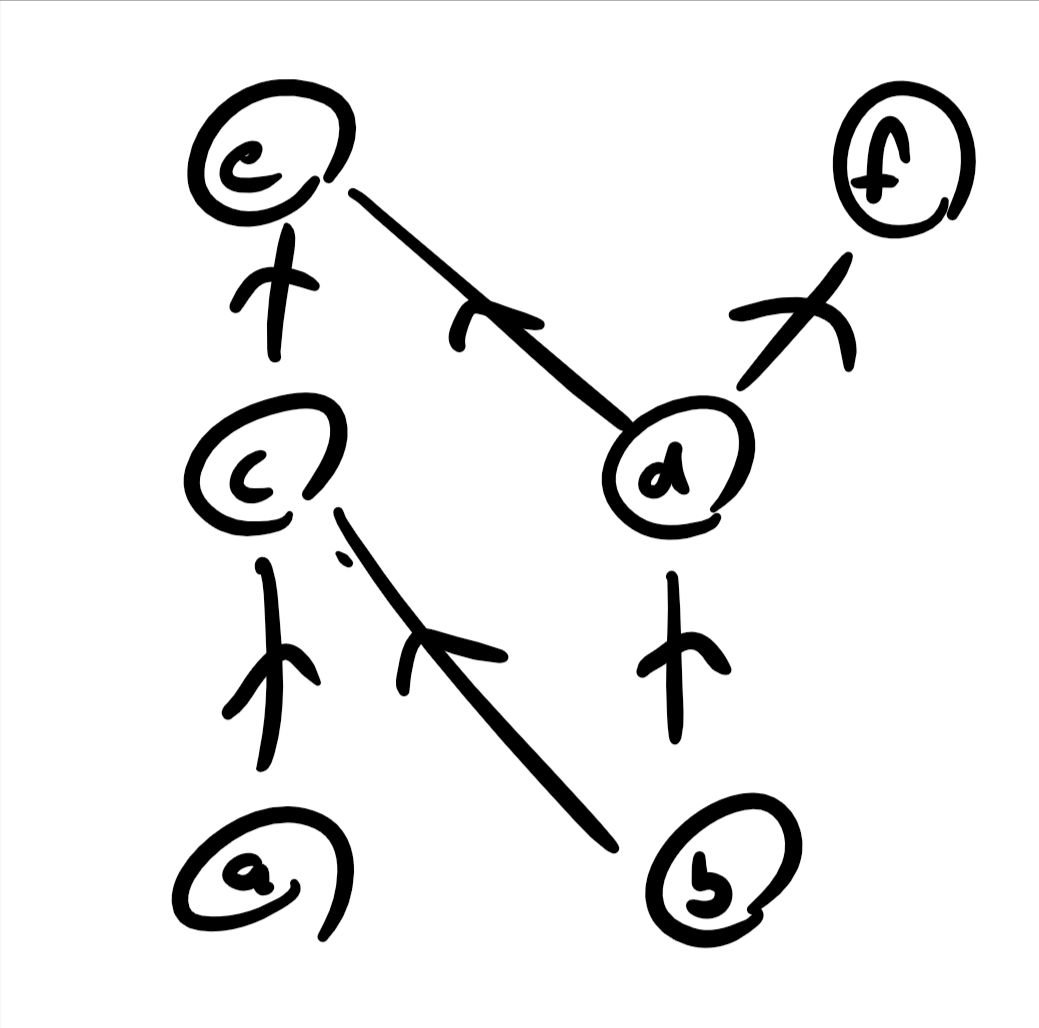}
        \hspace{5mm}
        \includegraphics[width=0.45\textwidth]{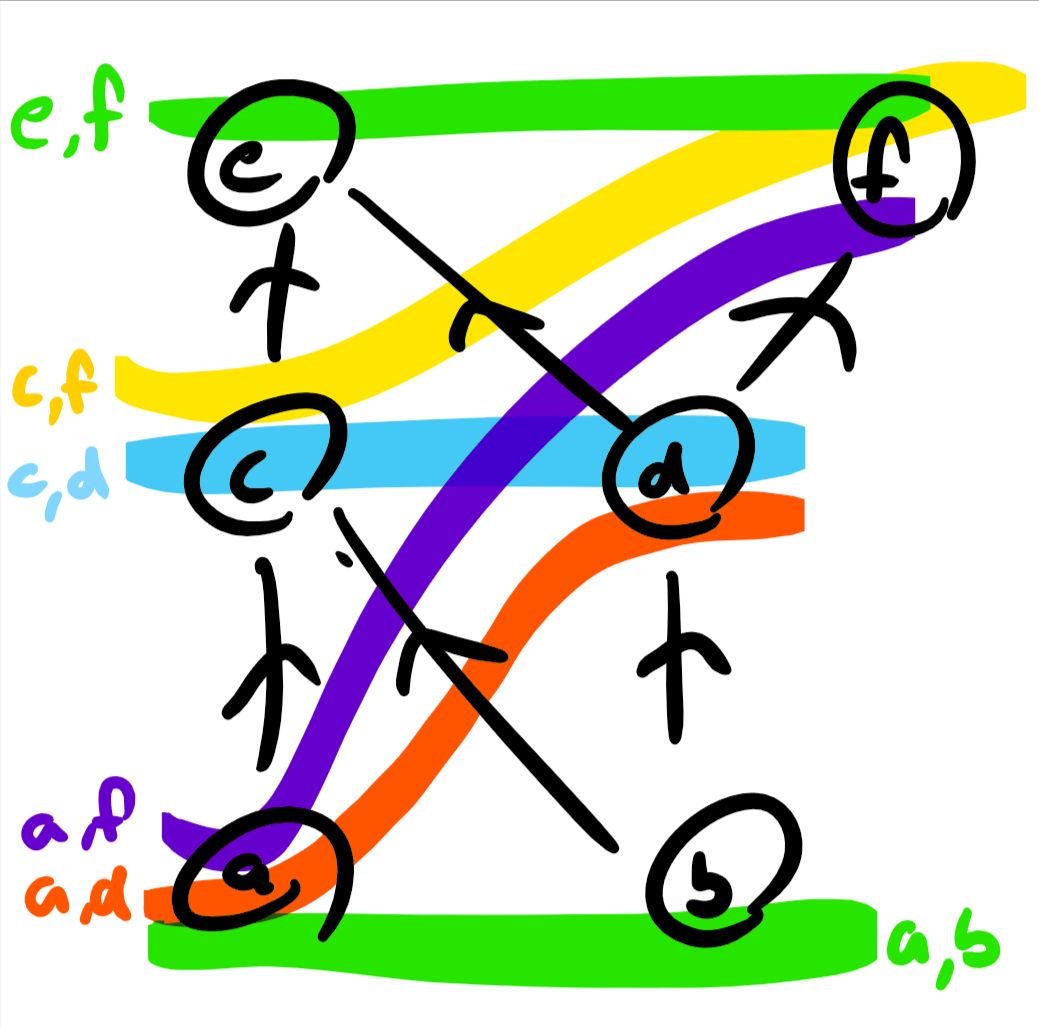}
    \end{center}
    \caption{Left: the Hasse diagram for a causal order on 6 events $\{a,b,c,d,e,f\}$. Right: the maximal slices for the causal order highlighted (all other slices can be obtained as subsets of the maximal slices).}
\end{figure}

\begin{definition}\label{definition:spacelike-separated}
    Let $\Omega$ be a causal order and let $\mathcal{A} \subseteq \powerset{\Omega}$ be a collection of subsets of $\Omega$. We say that the subsets in $\mathcal{A}$ are \emph{space-like separated} if the following conditions holds for all distinct $A, B\in \mathcal{A}$:
    \begin{equation}
        A \cap \left(\future{B} \cup \past{B}\right) = \emptyset
    \end{equation}
    In particular, a space-like slice is the union of a collection of space-like separated singleton subsets.
\end{definition}

\noindent More than diamonds or paths, slices are the focus of this work. Space-like slices are a generalisation of space-like surfaces from Relativity: the term ``slice'' is used here in place of ``surface'' because the latter traditionally implies some topological conditions.

\begin{definition}\label{definition:category-of-all-slices}
    Let $\Omega$ be a causal order. The \emph{category of all slices} on $\Omega$, denoted by $\sliceCat{\Omega}$, is the strict partially monoidal category \cite{gogioso2019process} defined as follows.
    \begin{itemize}
        \item Objects of $\sliceCat{\Omega}$ are the slices of $\Omega$.
        \item The category is a poset and the unique morphism from a space-like slice $\Sigma$ to another space-like slice $\Gamma$ is denoted $\Sigma \slicePreceq \Gamma$ if it exists.
        Specifically, we say that $\Sigma \slicePreceq \Gamma$ if and only if $\Gamma \subseteq \futuredom{\Sigma}$, i.e. iff $\Gamma$ lies entirely into the future domain of dependence of $\Sigma$.
        \item The monoidal product on objects $\Sigma \otimes \Gamma$ is only defined when $\Sigma$ and $\Gamma$ are space-like separated, in which case it is the disjoint union $\Sigma \sqcup \Gamma$.
        \item The unit for the monoidal product is the empty space-like slice $\emptyset \subseteq \Omega$.
        \item The partial monoidal product on objects extends to morphisms because whenever $\Sigma' \subseteq \futuredom{\Sigma}$ and $\Gamma' \subseteq \futuredom{\Gamma}$---i.e. whenever $\Sigma \slicePreceq \Sigma'$ and $\Gamma \slicePreceq \Gamma'$---we necessarily have:
        \begin{equation}
            \Sigma' \sqcup \Gamma'
            \subseteq
            \futuredom{\Sigma} \cup \futuredom{\Gamma}
            \subseteq
            \futuredom{\Sigma \sqcup \Gamma}
            \text{, i.e. }
            \Sigma \otimes \Gamma \slicePreceq \Sigma' \otimes \Gamma'
        \end{equation}
    \end{itemize}
    The partial monoidal product is strict, i.e. strictly associative and unital when all products are defined. The partial monoidal product is also commutative, i.e. it is symmetric (wherever defined) with an identity $\Sigma \otimes \Gamma = \Gamma \otimes \Sigma$ as the symmetry isomorphism.
\end{definition}

The order relation $\Sigma \slicePreceq \Gamma$ on slices has been defined in such a way as to ensure that the field state local to the the codomain slice $\Gamma$ will be entirely determined by evolution and marginalisation of the field state on the domain slice $\Sigma$. In particular, the definition is such that any sub-slice $\Sigma' \subseteq \Sigma$ necessarily satisfies $\Sigma \slicePreceq \Sigma'$, since the field state on $\Sigma'$ can be obtained from the field state on $\Sigma$ by marginalisation/discarding. The connection to marginalisation will be discussed in further detail in Subsection~\ref{subsection:causal-field-theories-no-signalling} below.

\begin{figure}[h]
    \begin{center}
        \includegraphics[width=0.45\textwidth]{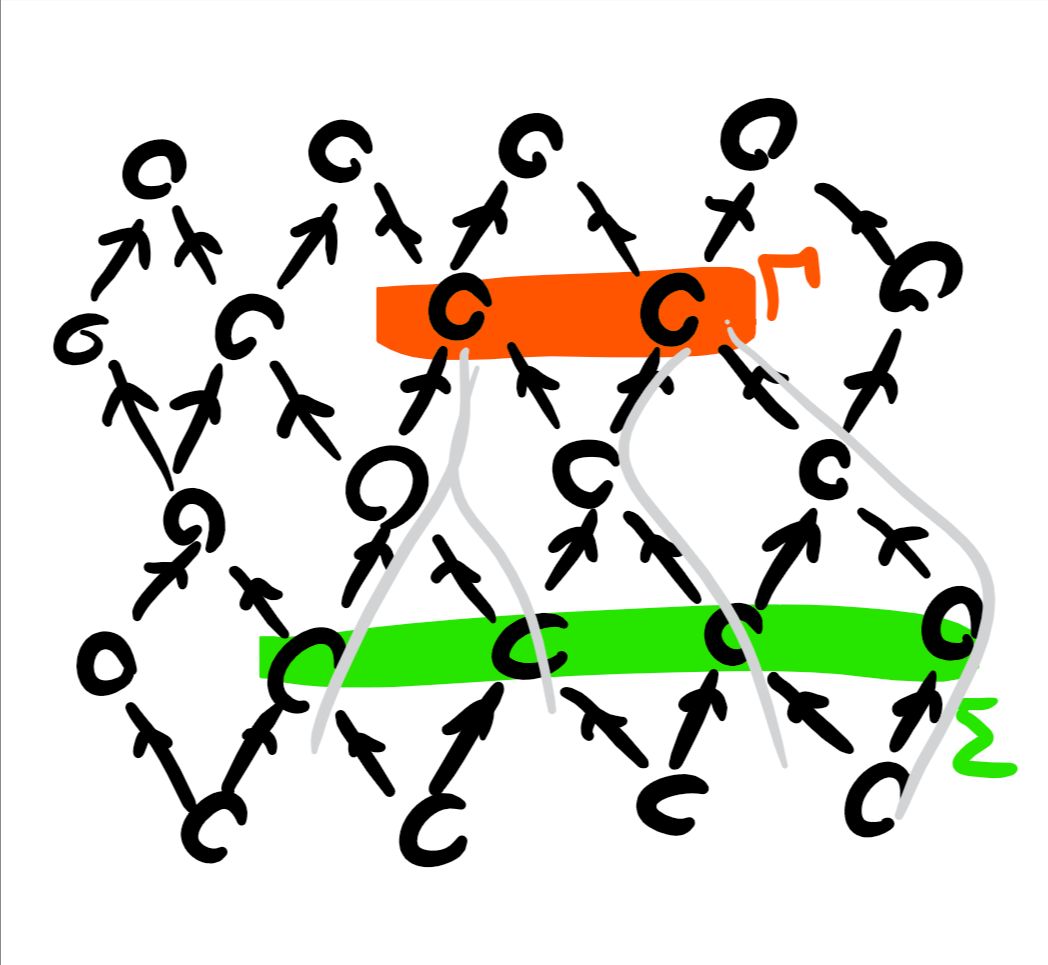}
        \hspace{5mm}
        \includegraphics[width=0.45\textwidth]{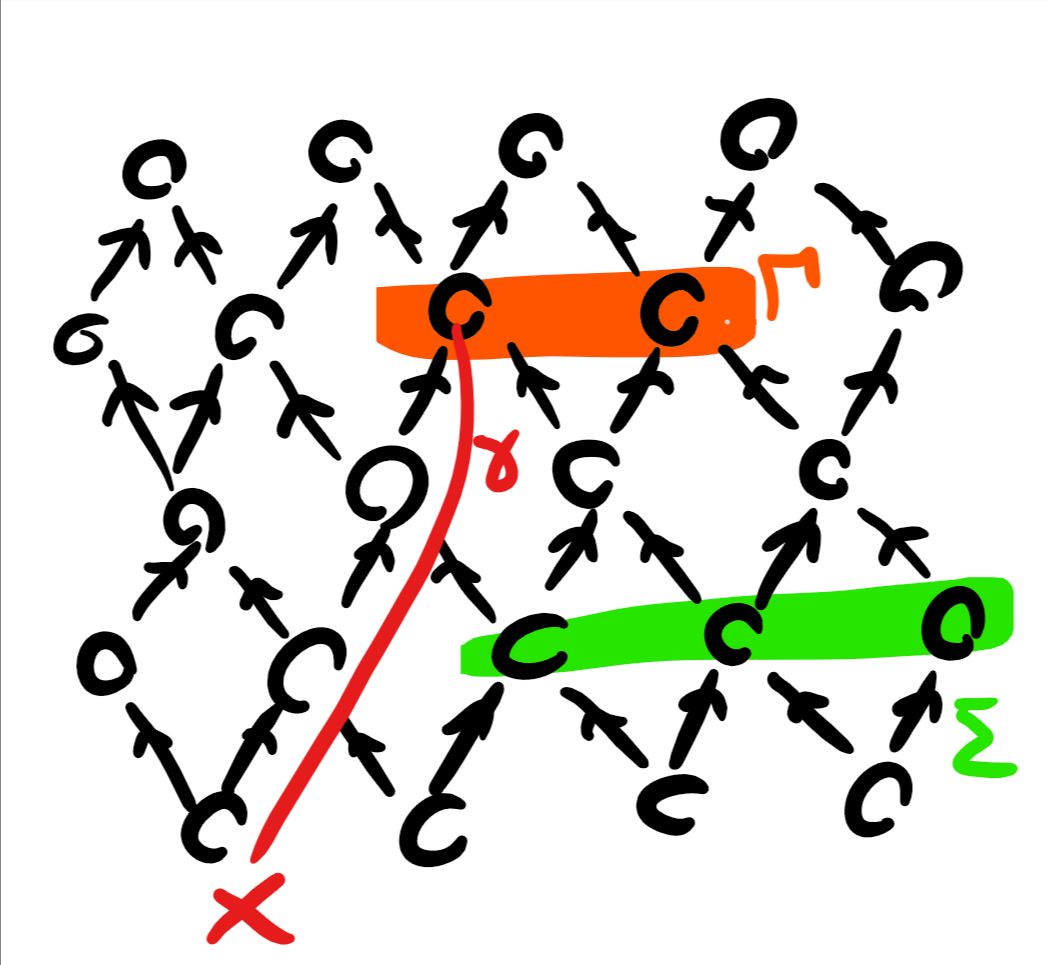}
    \end{center}
    \caption{Left: two slices $\Sigma, \Gamma$ such that $\Sigma \slicePreceq \Gamma$. Right: two slices $\Sigma, \Gamma$ such that $\Sigma \not\slicePreceq \Gamma$, highlighting a past-directed path $\gamma$ starting from an event of $\Gamma$ and not intersecting $\Sigma$ at any point.}
\end{figure}

\begin{figure}[h]
    \begin{center}
        \includegraphics[width=0.30\textwidth]{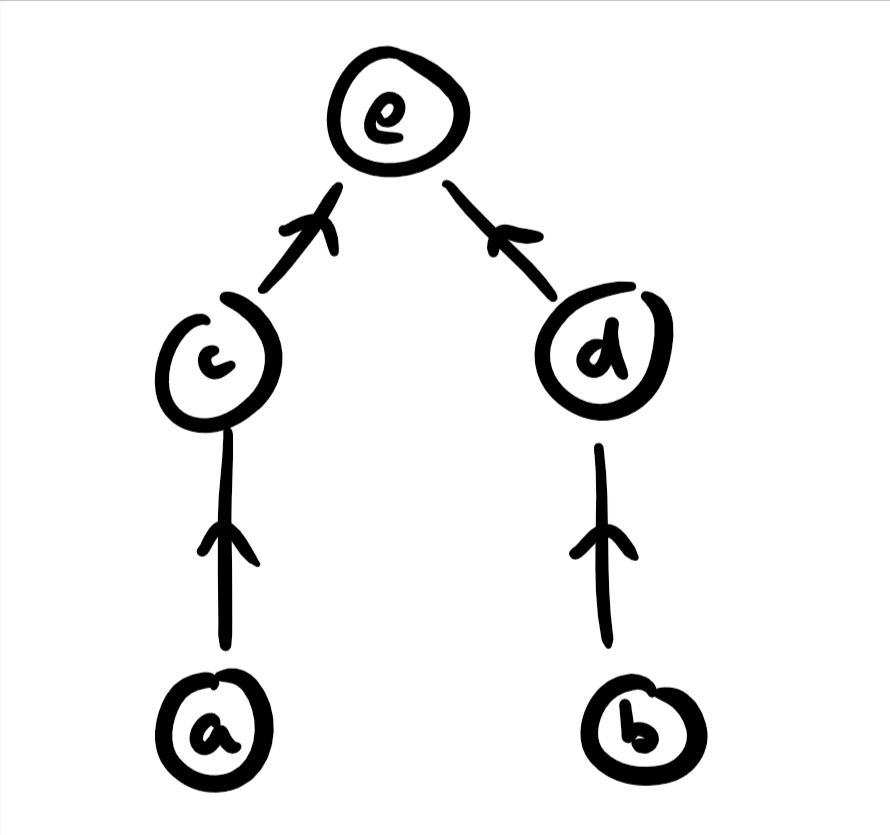}
        \hspace{2mm}
        \includegraphics[width=0.30\textwidth]{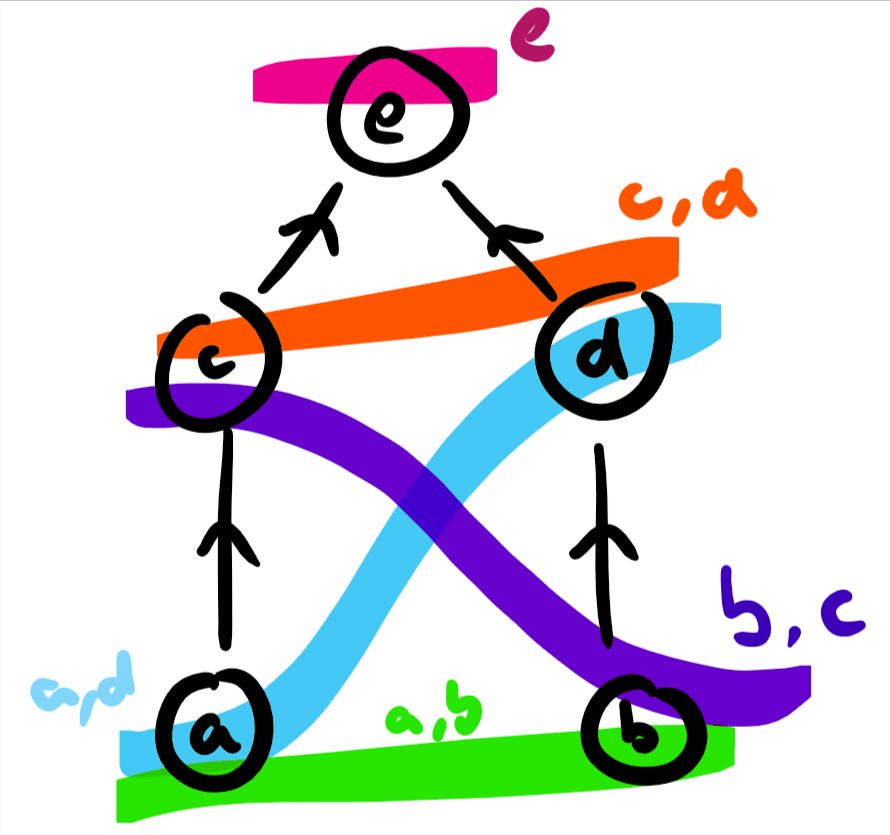}
        \hspace{2mm}
        \includegraphics[width=0.35\textwidth]{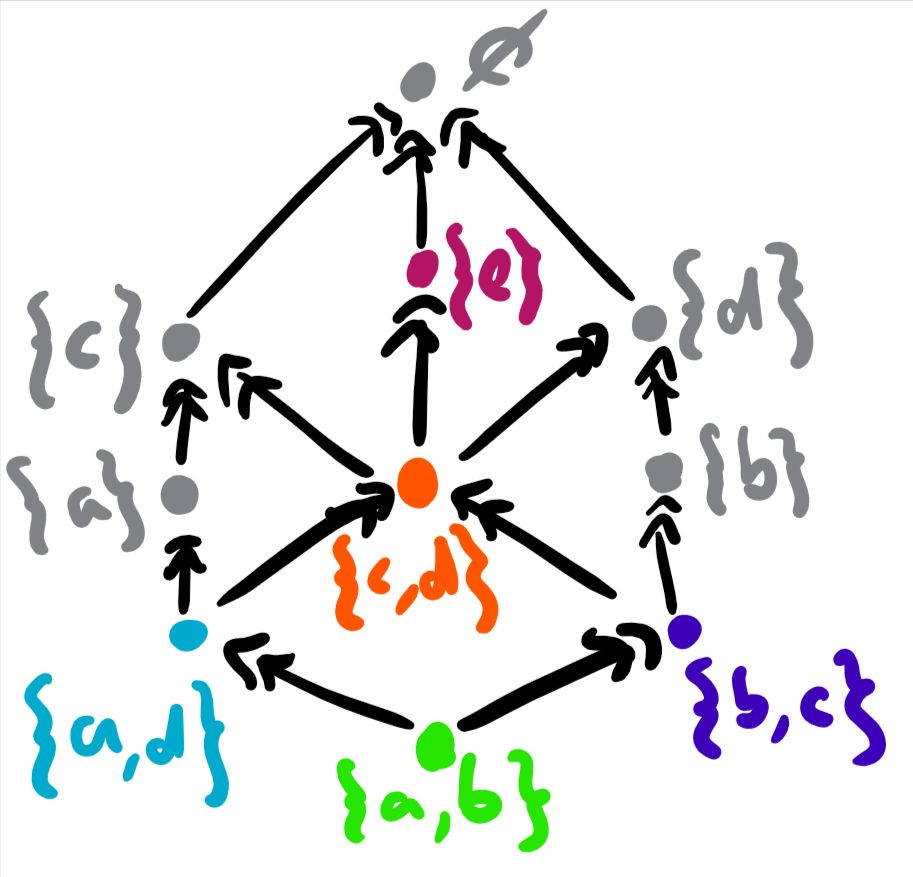}
    \end{center}
    \caption{Left: the Hasse diagram for a causal order. Centre: the maximal slices for the causal order highlighted. Right: the category of all slices for the causal order.}
\end{figure}

\subsection{Diamonds and Regions}

\begin{definition}\label{definition:causal-diamond}
    Let $\Omega$ be a causal order. If $x, y$ are two events in $\Omega$, the \emph{causal diamond} from $x$ to $y$ in $\Omega$ is the causal sub-order $(\Diamond_{x,y}, \leq) \hookrightarrow \Omega$ defined as follows:
    \begin{equation}\label{equation:causal-diamond}
        \Diamond_{x,y} := \suchthat{z \in \Omega}{x \leq z \leq y} = \bigcup_{\gamma: x \rightsquigarrow y} \!\gamma
    \end{equation}
\end{definition}

\begin{definition}\label{definition:region}
    Let $\Omega$ be a causal order. A \emph{region} in $\Omega$ is a causal sub-order $(R, \leq) \hookrightarrow \Omega$ which is \emph{convex}, i.e. one such that for all events $x, y \in R$ the causal diamond from $x$ to $y$ in $\Omega$ is a subset of $R$ (i.e. $R$ contains all paths $\gamma: x \rightsquigarrow y$ in $\Omega$).
\end{definition}

\noindent Definition~\ref{definition:region} is the order-theoretic incarnation of the requirement that causal diamonds generate the topology of Lorentzian manifolds: we could have equivalently stated it as saying that regions in $\Omega$ are all the possibly unions of causal diamonds in $\Omega$ (including the empty one). A special case of region of particular interest is the region \emph{between} two slices $\Sigma \slicePreceq \Gamma$.

\begin{definition}\label{definition:region-between-slices}
    Let $\Omega$ be a causal order and consider two slices $\Sigma \slicePreceq \Gamma$. We define the \emph{region between $\Sigma$ and $\Gamma$} as follows:
    \begin{equation}\label{equation:region-between-slices}
        \Diamond_{\Sigma,\Gamma} := \bigcup_{x \in \Sigma} \bigcup_{y \in \Gamma} \Diamond_{x,y}
    \end{equation}
    In particular, a causal diamond $\Diamond_{x, y}$ is the region between the slices $\{x\}$ and $\{y\}$. More generally, a region between slices $\Sigma$ and $\Gamma$ is the intersection $\Diamond_{\Sigma, \Gamma} = \future{\Sigma} \cap \past{\Gamma}$ of their future and past respectively.
\end{definition}

\noindent The slices $\Sigma$ and $\Gamma$ bounding the region $\Diamond_{\Sigma,\Gamma}$ can be obtained respectively as the sets of its minima $\Sigma = \min \Diamond_{\Sigma,\Gamma}$ and of its maxima $\Gamma = \max \Diamond_{\Sigma,\Gamma}$.
As a special case, a slice $\Sigma$ is the region between $\Sigma$ and $\Sigma$.
Conversely, every closed bounded region $R$---and in particular every finite region---is in the form $R = \Diamond_{\min R, \max R}$.

\begin{figure}[h]
    \begin{center}
        \includegraphics[width=0.45\textwidth]{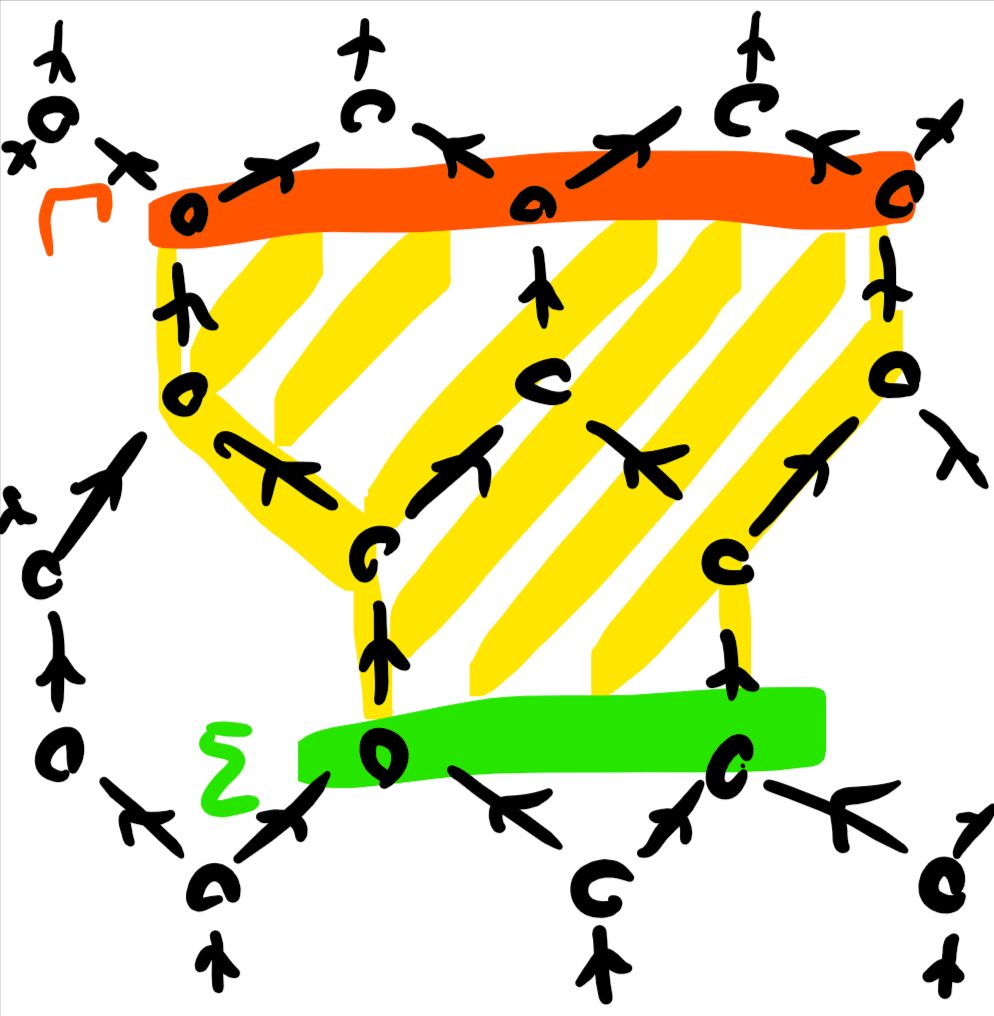}
        \hspace{5mm}
        \includegraphics[width=0.45\textwidth]{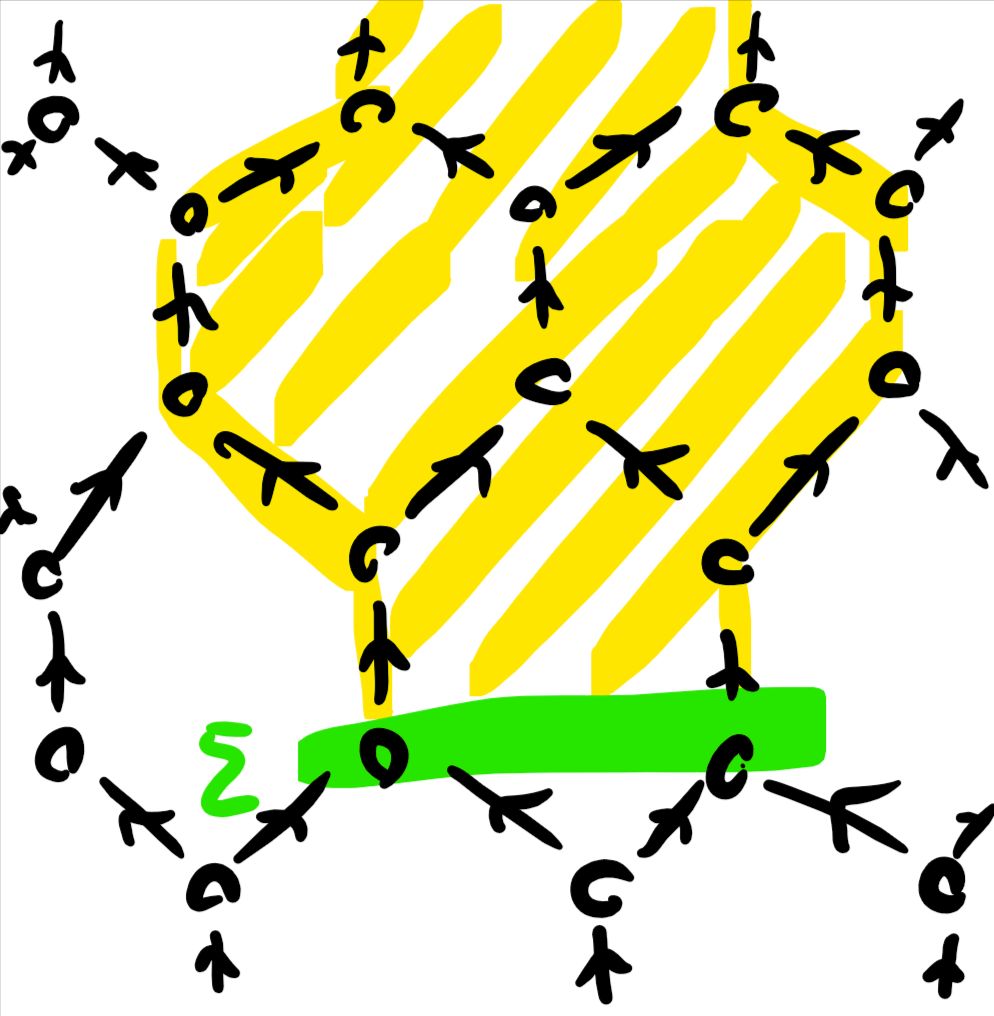}
    \end{center}
    \caption{Left: the region between two slices on the honeycomb lattice. Right: an unbounded (necessarily infinite) region on the honeycomb lattice.}
\end{figure}

\section{Categories of slices}
\label{section:categories-slices}

Because we didn't impose any topological constraints on the slices, it is possible that the category $\sliceCat{\Omega}$ will, in practice, contain objects which are to irregular or exotic for physical fields to be defined over (such as fractal slices with low topological dimension).
To obviate this issue, we consider more general categories of slices on a given causal order: this will allow us to restrict our attention to slices with any properties we desire, as long as we retain enough slices to reconstruct the structure of the causal order $\Omega$, both (1) globally and (2) locally.

No requirement is made for all products to exist in $\sliceCat{\Omega}$ to also exist on members of a more general category of slices: it is the case that certain properties desirable in practice may not be closed under arbitrary union of space-like separated slices themselves satisfying the property.
\footnote{An example of this phenomenon is given by constant-time partial Cauchy slices in Minkowski spacetime: the union of two disjoint constant-time partial Cauchy slices having the same time parameter yields another constant-time partial Cauchy slice, but the union of two space-like separated constant-time partial Cauchy slices having different time parameters does not yield a constant-time partial Cauchy slice as a result.}
However, we impose the requirement (3) that these more general categories of slices be partially monoidal sub-categories of $\sliceCat{\Omega}$.

\begin{definition}\label{definition:categories-of-slices}
    Let $\Omega$ be a causal order. A \emph{category of slices} on $\Omega$ is the full sub-category $\mathcal{C}$ of $\sliceCat{\Omega}$ defined by a given set $\obj{\mathcal{C}}$ of slices chosen in such a way that the following three conditions hold.
    \begin{enumerate}
        \item[(1)] For any two events $x, y \in \Omega$ with $x \leq y$, there exist slices $\Sigma, \Gamma \in \obj{\mathcal{C}}$ such that $x \in \Sigma$, $y \in \Gamma$ and $\Sigma \slicePreceq \Gamma$.
        \item[(2)] If $\Sigma, \Gamma$ and $\Delta$ are three slices in $\mathcal{C}$, then the restriction $(\Delta \cap \Diamond_{\Sigma, \Gamma})$ of $\Delta$ to the region $\Diamond_{\Sigma, \Gamma}$ is also a slice in $\mathcal{C}$.
        \item[(3)] The category of slices $\mathcal{C}$ is a partially monoidal subcategory of $\sliceCat{\Omega}$. In particular, $\emptyset \in \obj{\mathcal{C}}$ and whenever $\Sigma \otimes \Gamma$ exists in $\mathcal{C}$ for some $\Sigma, \Gamma \in \obj{\mathcal{C}}$ then $\Sigma \otimes \Gamma$ also exists in $\sliceCat{\Omega}$. (Associativity and unitality of $\otimes$ are strict in $\mathcal{C}$ as they are in $\sliceCat{\Omega}$.)
    \end{enumerate}
    In particular, $\sliceCat{\Omega}$ is itself a category of slices on $\Omega$.
\end{definition}

\noindent Condition (2) in the definition above tells us that we can talk about regions directly within a given category $\mathcal{C}$ of slices, without first having to reconstruct the causal order $\Omega$: this will form the basis of the connection to AQFT in Subsection~\ref{section:connection-aqft} below.

As an example of particularly well-behaved slices, we define a notion of \emph{Cauchy slices}---akin to that of Cauchy surfaces from Relativity---and remark that any ``foliation'' of a causal order in terms of such slices gives rise to what is arguably the simplest non-trivial example of category of slices.
\begin{definition}\label{definition:cauchy-slice}
    A slice $\Sigma$ on $\Omega$ is a \emph{Cauchy slice} if every causal path $\gamma: -\infty \rightsquigarrow +\infty$ in $\Omega$ intersects $\Sigma$ at some (necessarily unique) event. Cauchy slices are in particular maximal slices.
    A \emph{category of Cauchy slices} on $\Omega$ is a category $\mathcal{C}$ of slices on $\Omega$ such that every slice $\Sigma \in \obj{\mathcal{C}}$ is a subset $\Sigma \subseteq \Gamma$ of some Cauchy slice $\Gamma \in \obj{\mathcal{C}}$.
\end{definition}
\begin{proposition}\label{proposition:cauchy-slice-category}
    A \emph{foliation} on a causal order $\Omega$ is a set $\mathcal{F}$ of Cauchy slices on $\Omega$ such that:
    \begin{enumerate}
        \item[(1)] the slices in $\mathcal{F}$ are totally ordered according to $\slicePreceq$;
        \item[(2)] every event $x \in \Omega$ is contained in some slice $\Sigma \in \mathcal{F}$;
        \item[(3)] the slices in $\mathcal{F}$ are pairwise disjoint.
    \end{enumerate}
    If $\mathcal{F}$ is a foliation, write $\cauchySliceCat{\mathcal{F}}$ for the full sub-category of $\sliceCat{\Omega}$ generated by all slices which are subsets of some Cauchy slice in $\mathcal{F}$.
    Then $\cauchySliceCat{\mathcal{F}}$ is a category of Cauchy slices on $\Omega$.
\end{proposition}
\begin{proof}
    Let $\cauchySliceCat{\mathcal{F}}$ denote the full sub-category of  $\sliceCat{\Omega}$ generated by all slices which are subsets of some Cauchy slice.

    For any two events $x \leq y$ in $\Omega$, let $\Sigma, \Gamma \in \obj{\cauchySliceCat{\mathcal{F}}}$ be two Cauchy slices such that $x \in \Sigma$ and $y \in \Gamma$, the existence of such slices guaranteed by the definition of foliation.
    Because the foliation is totally ordered, we have that $\Sigma \slicePreceq \Gamma$ or $\Gamma \slicePreceq \Sigma$ (or both, if $\Sigma = \Gamma$ and $x = y$).
    If $x = y$, either works, while if $x < y$ then necessarily $\Sigma \slicePreceq \Gamma$. Either way, condition (1) for $\cauchySliceCat{\mathcal{F}}$ to be a category of slices is satisfied.

    Let $\Sigma'$, $\Gamma'$ and $\Delta'$ be three slices, respectively contained in three Cauchy slices $\Sigma$, $\Gamma$ and $\Delta$ inside the foliation. Because of total ordering and disjointness of slices in $\mathcal{F}$, the only instance in which $\Delta \cap \Diamond_{\Sigma, \Gamma} \neq \emptyset$ is when $\Sigma \slicePreceq \Delta \slicePreceq \Gamma$. In this case, $\Delta \cap \Diamond_{\Sigma, \Gamma} = \Delta \in \obj{\cauchySliceCat{\mathcal{F}}}$. Otherwise, $\Delta \cap \Diamond_{\Sigma, \Gamma} = \emptyset \in \obj{\cauchySliceCat{\mathcal{F}}}$. Either way, condition (2) for $\cauchySliceCat{\mathcal{F}}$ to be a category of slices is satisfied when $\Sigma$, $\Gamma$ and $\Delta$ are Cauchy slices. This result immediately generalises to $\Sigma'$, $\Gamma'$ and $\Delta'$: we have that $\Delta' \cap \Diamond_{\Sigma', \Gamma'} \subseteq \Delta \cap \Diamond_{\Sigma, \Gamma} \subseteq \Delta$, so that $\Delta' \in \obj{\cauchySliceCat{\mathcal{F}}}$ and condition (2) for $\cauchySliceCat{\mathcal{F}}$ to be a category of slices is satisfied.

    Finally, if $\Sigma, \Gamma$ are two slices such that $\Sigma \otimes \Gamma$ is defined in $\cauchySliceCat{\mathcal{F}}$, then $\Sigma, \Gamma$ are necessarily disjoint subsets of the same Cauchy slice $\Delta$. It is then immediate to conclude that condition (3) for $\cauchySliceCat{\mathcal{F}}$ to be a category of slices is satisfied.
\hfill$\square$\end{proof}

\subsection{The category of causal orders}

As objects, causal orders have been defined simply as posets. However, causal orders should not relate to each other simply as posets: Malament's result \cite{malament1977class} may seem at first to indicate that order-preserving maps are the correct choice, but upon closer inspection one realises that the result itself only talks about order-preserving \emph{isomorphisms}, giving no indication about other maps.

A prototypical example of the behaviour we wish to avoid is that where $\Omega' \hookrightarrow \Omega$ is a sub-poset such that $x \leq y$ in $\Omega$ for some $x, y \in \Omega'$ but $x \not{\!\!\leq} y$ in $\Omega'$.
The issue above is the reason behind the rather specific formulation of the notion of \emph{causal sub-order} in Definition~\ref{definition:causal-order}, prompting us to choose a special subclass of order-preserving maps as morphisms between causal orders.

\begin{definition}\label{definition:causal-orders-category}
    The category $\causOrdCat$ of \emph{causal orders} is the symmetric monoidal category defined as follows:
    \begin{itemize}
        \item Objects of $\causOrdCat$ are causal orders, i.e. posets.
        \item Morphisms $\Omega \rightarrow \Theta$ in $\causOrdCat$ are the order-preserving functions $f: \Omega \rightarrow \Theta$ such that $x < y$ in $\Omega$ whenever $f(x) < f(y)$ in $\Theta$.
        \item The monoidal product on objects $\Omega \otimes \Theta$ is the (forcedly) disjoint union $\Omega \sqcup \Theta := \Omega \times \{0\} \cup \Theta \times \{1\}$.
        \item The unit for the monoidal product is the empty causal order $\emptyset$.
        \item The monoidal product extends to the disjoint union of morphisms. If $f: \Omega \rightarrow \Omega'$ and $g: \Theta \rightarrow \Theta'$, then the monoidal product $f \otimes g: \Omega \otimes \Theta \rightarrow \Omega' \otimes \Theta'$ is defined as follows:
        \begin{equation}
            f \otimes g := f \sqcup g
            =
            (x,i) \mapsto
            \begin{cases}
                (f(x),0) \in \Omega'\times \{0\} &\text{ if } i = 0\\
                (g(x),1) \in \Theta'\times \{1\} &\text{ if } i = 1
            \end{cases}
        \end{equation}
    \end{itemize}
    The monoidal product is not strict nor commutative, but symmetric under the symmetry isomorphisms $s: \Omega \sqcup \Theta \rightarrow \Theta \sqcup \Omega$ defined by $s(x, i) = (x, 1-i)$.
\end{definition}

It is easy to check that the causal sub-orders $\Omega'$ of a causal order $\Omega$ according to Definition~\ref{definition:causal-order} are all sub-objects $\Omega' \hookrightarrow \Omega$ in the category $\causOrdCat$, so that the notion of causal sub-order is consistent with the usual notion of categorical sub-object.
As discussed above, the regions in a causal order $\Omega$ are examples of causal sub-orders, but not all sub-orders are regions: e.g. paths are always sub-orders but not necessarily regions.
In general, if we have $\Omega' \hookrightarrow \Omega$ then it is not necessary for $\Omega'$ to be convex, i.e. it is not necessary for $\Omega'$ to contain all paths $x \rightsquigarrow y$ in $\Omega$ for any two events $x,y \in \Omega'$: in a sense, $\Omega$ can ``refine'' the causal sub-order $\Omega'$ by adding new events in between events of the latter.
As the following proposition shows, this ``refinement'' of causal orders is the only other case we need to consider when talking about causal sub-orders.

\begin{definition}\label{definition:regions-and-refinements}
    Let $\Omega$ be a causal order and let $i: \Omega' \hookrightarrow \Omega$ be a causal sub-order of $\Omega$.
    We say that $i: \Omega' \hookrightarrow \Omega$ is a \emph{region} if the image $i(\Omega') \subseteq \Omega$ is a region in $\Omega$.
    We say that $i: \Omega' \hookrightarrow \Omega$ is a \emph{refinement} if the image $i(\Omega') \subseteq \Omega$ is such that for all $x \leq y \in \Omega$ there exist $x', y' \in \Omega'$ with $i(x') \leq x \leq y \leq i(y')$.
\end{definition}

\begin{proposition}\label{proposition:regions-and-refinements}
    Let $\Omega$ be a causal order and let $i: \Omega' \hookrightarrow \Omega$ be a causal sub-order of $\Omega$. Then $i$ factors (essentially) uniquely as $i = r \circ f$ for some region $r: \Theta \hookrightarrow \Omega$ and some refinement $f: \Omega' \hookrightarrow \Theta$.
\end{proposition}
\begin{proof}
    Let $\Theta$ be the region of $\Omega$ obtained as the union of the causal diamonds $\Diamond_{x,y}$ for all $x,y \in i(\Omega')$. Let $r: \Theta \hookrightarrow \Omega$ be the injection of $\Theta$ into $\Omega$ as a sub-poset and let $f: \Omega' \hookrightarrow \Theta$ be the restriction of the codomain of $i$ to $\Theta$: clearly $i = r \circ f$, $r$ is a region and $f$ is a refinement (because of how $\Theta$ was constructed).

    Now let $\Theta'$ be such that  $r': \Theta' \hookrightarrow \Omega$ is a region and $f': \Omega' \hookrightarrow \Theta'$ is a refinement with $i = r' \circ f'$: to prove essential uniqueness, we want to show that there is some isomorphism $\theta: \Theta' \rightarrow \Theta$ such that $f = \theta \circ f'$ and $r' = r \circ \theta$.
    Because $r' \circ f' = i$, the image $r'(\Theta')$ is the region $\Theta$ itself, so that the restriction $\theta: \Theta' \rightarrow \Theta$ of the codomain of $r'$ to $\Theta$ is an isomorphism with $r' = r \circ \theta$.
    Now we have $r \circ f = i = r' \circ f' = r \circ (\theta \circ f')$: but $r$ is a monomorphism (i.e. it is injective), so necessarily $f = \theta \circ f'$.
\hfill$\square$\end{proof}

\noindent The category $\causOrdCat$ also has epi-mono factorisation, i.e. every morphism $f: \Omega' \rightarrow \Omega$ can be factorised (essentially) uniquely as an epimorphism (i.e. a surjective map) $q: \Omega' \rightarrow \Theta$ and a monomorphism (i.e. an injective map) $i: \Theta \hookrightarrow \Omega$. We have already adopted the nomenclature \emph{causal sub-order} for the latter form of morphism, while we will henceforth use \emph{causal quotient} to refer to the former.
A snippet of the the causal quotient $q: H \rightarrow D$ from the (infinite) honeycomb lattice to the (infinite) diamond lattice is shown in Figure~\ref{figure:honeycomb-diamond-lattices}.

\begin{figure}[h]
    \begin{center}
        \includegraphics[width=0.95\textwidth]{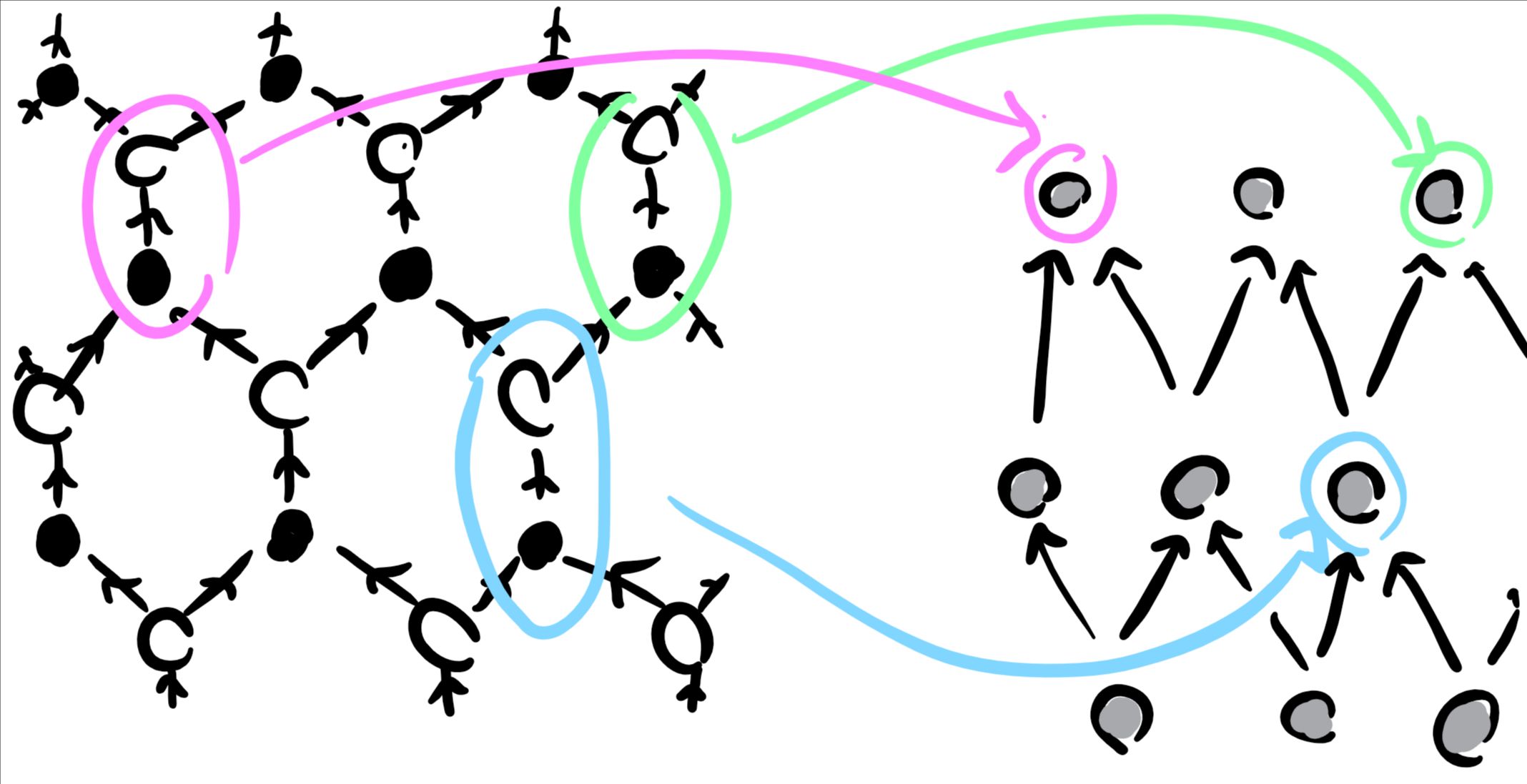}
    \end{center}
    \caption{Causal quotient from the honeycomb lattice to the diamond lattice. The pre-images of three events from the diamond lattice are highlighted.}
    \label{figure:honeycomb-diamond-lattices}
\end{figure}

If $\Sigma$ is a slice in $\Omega$, we can define its \emph{pullback} $f^\ast(\Sigma)$ to be the causal suborder of $\Theta$ generated by $\suchthat{x \in \Theta}{f(x) \in \Sigma}$, i.e. largest causal sub-order of $\Theta$ mapped onto $\Sigma$.
The pullback of a slice $\Sigma$ has a rather simple structure: the slices $\Gamma$ in the pullback $f^\ast(\Sigma)$ are exactly the disjoint unions $\Gamma := \bigsqcup_{x \in \Sigma} \Gamma_x$ for all possible choices $(\Gamma_x)_{x \in \Sigma}$ of slice sections of $f$ over the individual events $x$ of $\Sigma$.
\begin{equation}\label{equation:slice-sections}
    (\Gamma_x)_{x \in \Sigma}
    \in
    \prod_{x \in \Sigma}
    \obj{\sliceCat{f^\ast(\{x\})}}
\end{equation}

A depiction of the pullback under the causal quotient $q: H \rightarrow D$ described above can be seen in Figure~\ref{figure:honeycomb-diamond-lattices-preimages}.

\noindent If $\mathcal{C}$ is a category of slices on $\Omega$, we can define its \emph{pullback} along $f$ to be the full sub-category $f^\ast(\mathcal{C})$ of $\sliceCat{\Theta}$ spanned by all slices $\Gamma$ in $\Theta$ such that $\Gamma \in \obj{\sliceCat{f^\ast(\Sigma)}}$ for some $\Sigma \in \obj{\mathcal{C}}$.
The relationship $\slicePreceq$ between slices in pullbacks is a little complicated and its full characterisation is left to future work.

\begin{remark}\label{remark:pullback-slices}
    The two notions of pullback defined above---for slices and for categories of slices---are related by the observation that $\sliceCat{f^\ast(\Sigma)} = f^\ast(\sliceCat{\Sigma})$ for any slice $\Sigma$ of $\Omega$ (which can equivalently be seen as a causal sub-order $\Sigma \hookrightarrow \Omega$).
\end{remark}

\begin{figure}[h]
    \begin{center}
        \includegraphics[width=0.95\textwidth]{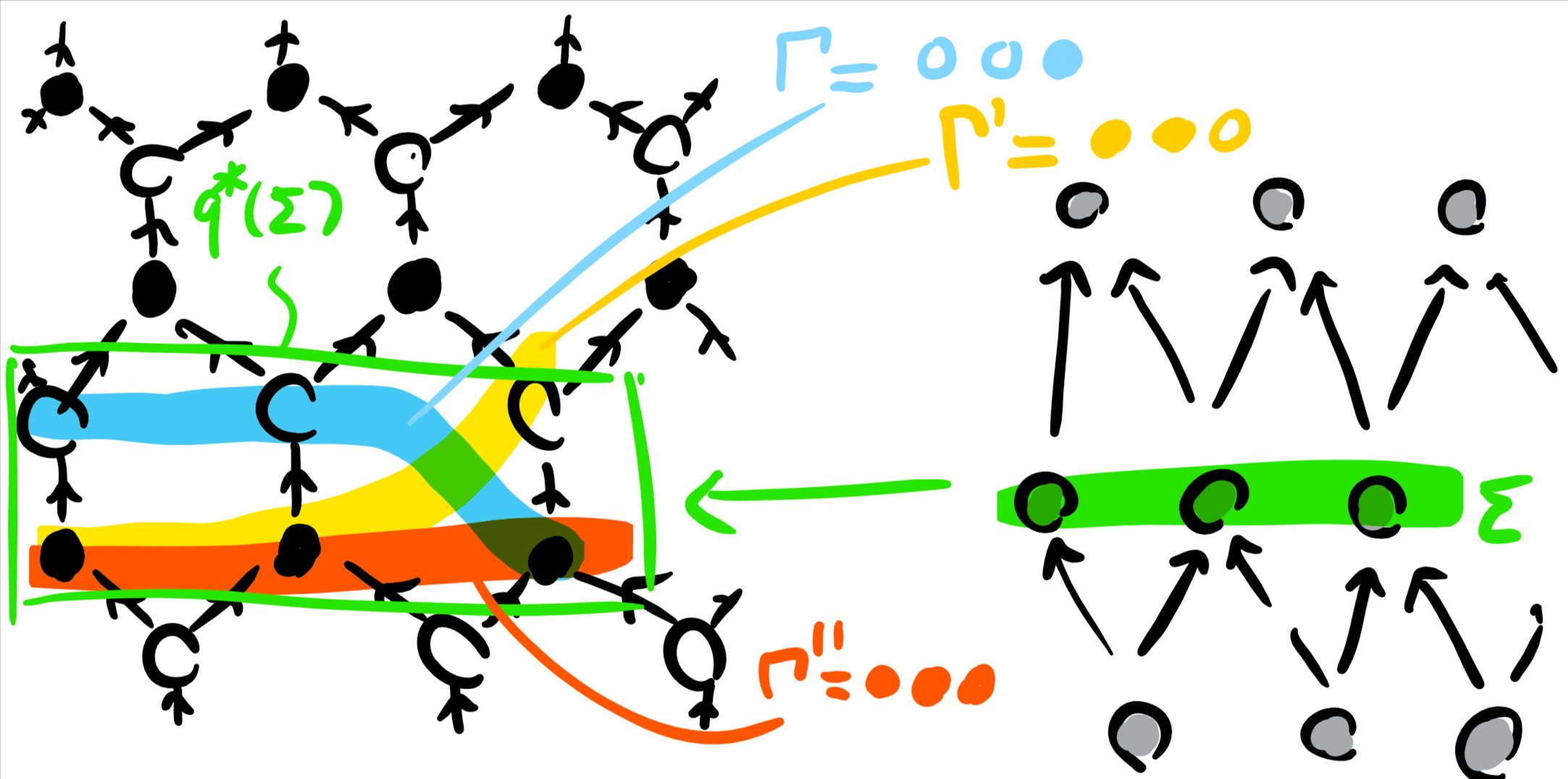}
    \end{center}
    \caption{A slice $\Sigma$ on the diamond lattice and three slices $\Gamma$, $\Gamma'$ and $\Gamma''$ in its pullback $q^\ast(\Sigma)$ on the honeycomb lattice.}
    \label{figure:honeycomb-diamond-lattices-preimages}
\end{figure}

\section{Causal field theories}

In the previous Section, we have defined several commonplace notions from Relativity in the more abstract context of causal orders.
In this Section, we endow our causal order with fields, living in an appropriate symmetric monoidal category.

\subsection{Categories for quantum fields}
\label{subsection:categories-for-quantum-fields}

Depending on the specific applications, there are many symmetric monoidal categories available to model quantum fields.

\begin{itemize}
    \item If the context is finite-dimensional, quantum fields can be taken to live in the category $\cpmCat{\fhilbCat}$ of finite-dimensional Hilbert spaces and completely positive maps between them.
    \item If the context is finite-dimensional and super-selected systems are of interest, quantum fields can be taken to live in the category $\cpstarCat{\fhilbCat} \cong \fcstaralgCat$ of finite-dimensional C*-algebras and completely positive maps between them.
    \item If the context is finite-dimensional, an even richer playground available for quantum fields is the category $\splitCat{\cpmCat{\fhilbCat}}$: this is the Karoubi envelope of the category $\cpmCat{\fhilbCat}$, containing $\cpstarCat{\fhilbCat}$ and a number of other systems of operational interest (such as fixed-state systems and constrained systems, see e.g. \cite{gogioso2017categorical}).
    \item If the context is infinite-dimensional, e.g. in the case of AQFT \cite{halvorson2006algebraic,heunen2009topos}, the categories usually considered for quantum fields are the category $\hilbCat$ of Hilbert spaces and bounded linear maps, the category $\cstaralgCat$ of C*-algebras and its subcategories $\wstaralgCat$ of W*-algebras (sometimes known as ``abstract'' von Neumann algebras) and $\vnalgCat$ of (concrete) von Neumann algebras.
    \item The categories $\hilbCat$, $\cstaralgCat$, $\wstaralgCat$ and $\vnalgCat$ have some annoying limitations, so in an infinite-dimensional context one can alternatively work with hyperfinite quantum systems \cite{gogioso2018quantum}, which incorporate infinities and infinitesimals to offer additional features---such as duals, traces and unital Frobenius algebras---over plain Hilbert spaces and C*-algebras.
\end{itemize}

\noindent The framework we present here is agnostic to the specific choice of process theory (aka symmetric monoidal category) for quantum fields. In fact, it is agnostic to the specific physical theory considered for the fields: any causal process theory can be considered.

\subsection{Causal field theories}

\begin{definition}\label{definition:causal-quantum-field-theory}
    Let $\Omega$ be a causal order. A \emph{causal field theory} $\Psi$ on $\Omega$ is a monoidal functor $\Psi: \mathcal{C} \rightarrow \mathcal{D}$ from a category $\mathcal{C}$ of slices on $\Omega$ to some symmetric monoidal category $\mathcal{D}$, which we refer to as the \emph{field category}.
\end{definition}

\begin{remark}
        It may sometimes be desirable to add a requirement of injectivity on objects for the functor $\Psi$.
        This has two main motivations, one of physical character and one of mathematical character.
        Physically, injectivity means that the field spaces corresponding to distinct events have distinct identities (although they can be isomorphic).
        Mathematically, injectivity means that the image of the functor is itself a sub-category of $\mathcal{D}$, matching the style used by other works on compositional causality \cite{coecke2013causal,coecke2016terminality,gogioso2017categorical,pinzani2019categorical}.
        While we do not require this as part of our definition, we will take care for the constructions hereafter to be sufficiently general to accommodate the possibility that such a requirement be imposed.
\end{remark}

\noindent We now ask ourselves: what physical information does the functor $\Psi$ encode?
On objects, $\Psi$ associates each space-like slice $\Sigma$ to the space $\Psi(\Sigma)$ of fields over that slice: every point in $\Psi(\Sigma)$ is a valid initial condition for field evolution in the future domain of dependence for $\Sigma$.

\begin{remark}\label{remark:finite-slices}
    If $\Sigma$ is finite and the singleton slices $\{x\}$ for the individual events $x \in \Sigma$ are all in the chosen category $\mathcal{C}$ of slices, then the action of $\Psi$ on $\Sigma$ always factorises into the tensor product of its action on the individual events:
    \begin{equation}
        \Psi(\Sigma) = \bigotimes_{x \in \Sigma} \Psi(\{x\})
    \end{equation}
\end{remark}

\noindent On morphisms, $\Psi$ associates $\Sigma \slicePreceq \Gamma$ to a morphism $\Psi(\Sigma) \rightarrow \Psi(\Gamma)$: this is a specification of how the field evolves from $\Sigma$ to $\Gamma$, i.e. this is the map sending a field state $\ket{\phi}$ over the initial slice $\Psi(\Sigma)$ to the evolved field state $\Psi\left(\Sigma \slicePreceq \Gamma\right)\ket{\phi}$ over the final slice $\Psi(\Gamma)$. This identification of functorial action with field evolution is the core idea of our work.
In particular, it explains our specific definition of morphisms in $\sliceCat{\Omega}$, and hence in all categories of slices: $\Sigma \slicePreceq \Gamma$ if and only if the field data on $\Sigma$ is sufficient to derive the field data on $\Gamma$, assuming causal field evolution.
Monoidality of the functor on objects says that the space of fields on the union of disjoint slices is the monoidal product---the tensor product, when working in the familiar linear settings of Hilbert spaces, C*-algebras, von Neumann algebras, etc---of the spaces of fields on the individual slices.
Note that this requirement is stronger than the requirement imposed by AQFT, where field algebras over space-like separated diamonds are only required to commute as sub-algebras of the the global field algebra, not necessarily to take the form of a tensor product sub-algebra.

Functoriality and monoidality on morphisms have some interesting consequences, which we now discuss in detail.
Let $\Sigma \slicePreceq \Sigma'$ and $\Gamma \slicePreceq \Gamma'$ for a pair of space-like separated slices $\Sigma$ and $\Gamma$ and another pair of space-like separated slices $\Sigma'$ and $\Gamma'$.
Consider the field evolution between the two disjoint unions of slices:

\begin{equation}
    \Psi\big(
        (\Sigma \otimes \Gamma) \slicePreceq (\Sigma' \otimes \Gamma')
    \big)
    :\Psi(\Sigma) \otimes \Psi(\Gamma) \rightarrow \Psi(\Sigma') \otimes \Psi(\Gamma')
\end{equation}

\noindent Monoidality on morphisms implies that the field evolution above factors as the product of the individual field evolutions $\Psi(\Sigma) \rightarrow \Psi(\Sigma')$ and $\Psi(\Gamma) \rightarrow \Psi(\Gamma')$:

\begin{equation}
    \Psi\big(
        (\Sigma \otimes \Gamma) \slicePreceq (\Sigma' \otimes \Gamma')
    \big)
    =
    \Psi(\Sigma \slicePreceq \Sigma') \otimes \Psi(\Gamma \slicePreceq \Gamma')
\end{equation}

\noindent This may look surprising at first, but it becomes entirely natural upon observing the following.

\begin{proposition}\label{proposition:monoidal-product-slices}
    Let $\Omega$ be a causal order. If $\Sigma$ and $\Gamma$ are space-like separated slices in $\Omega$ and $\Sigma \slicePreceq \Sigma'$, then $\Sigma'$ and $\Gamma$ are also space-like separated slices.
\end{proposition}
\begin{proof}
    If $\Sigma$ and $\Gamma$ are space-like separated, then $\Gamma \cap (\future{\Sigma} \cup \past{\Sigma}) = \emptyset$. Because $\Sigma \slicePreceq \Sigma'$, furthermore, Proposition~\ref{proposition:domain-of-dependence-past-future-2} tells us that $\future{\Sigma'} \cup \past{\Sigma'} \subseteq \future{\Sigma} \cup \past{\Sigma}$. We conclude that $\Gamma \cap (\future{\Sigma'} \cup \past{\Sigma'}) = \emptyset$, i.e. that $\Sigma'$ and $\Gamma$ are also space-like separated.
\hfill$\square$\end{proof}

\noindent Proposition~\ref{proposition:monoidal-product-slices} above tells us that in our factorisation scenario the entire region between $\Sigma$ and $\Sigma'$ on one side and the entire region between $\Gamma$ and $\Gamma'$ on the other side are space-like separated.
Thus any causal field evolution from $\Sigma \otimes \Gamma$ to $\Sigma' \otimes \Gamma'$ would physically be expected to factor: this can be seen as a manifestation of the \emph{principle of locality} for field theories, sometimes also known as ``clustering''.

\subsection{Causality and no-signalling}
\label{subsection:causal-field-theories-no-signalling}

Because any category $\mathcal{C}$ of slices on a causal order $\Omega$ is a partially monoidal subcategory of $\sliceCat{\Omega}$, in particular it necessarily contains the empty slice $\empty$ (the monoidal unit).
We define the following family of effects, indexed by all slices $\Sigma \in \obj{\mathcal{C}}$:

\begin{equation}
    \trace{\Sigma} := \Psi(\Sigma \slicePreceq \emptyset)
\end{equation}

\noindent By monoidality we have that $\trace{\emptyset}$ is some scalar in the field category $\mathcal{D}$ and that the family respects the partial monoidal structure:

\begin{equation}
    \trace{\Sigma \otimes \Gamma}
    =
    \Psi\big((\Sigma \otimes \Gamma) \slicePreceq \emptyset\big)
    =
    \Psi\big((\Sigma \slicePreceq \emptyset) \otimes (\Gamma \slicePreceq \emptyset)\big)
    =
    \Psi(\Sigma \slicePreceq \emptyset) \otimes \Psi(\Gamma \slicePreceq \emptyset)
    =
    \trace{\Sigma} \otimes \trace{\Gamma}
\end{equation}

\noindent By functoriality, furthermore, the family of effects above is respected by the image of the functor:

\begin{equation}
    \trace{\Gamma} \circ \Psi(\Sigma \slicePreceq \Gamma)
    =
    \Psi(\Gamma \slicePreceq \emptyset) \circ \Psi(\Sigma \slicePreceq \Gamma)
    =
    \Psi(\Sigma \slicePreceq \emptyset)
    =
    \trace{\Sigma}
\end{equation}

\noindent This means that the family of effects $(\trace{\Sigma})_{\Sigma \in \obj{\mathcal{C}}}$ defined above is an \emph{environment structure} and that---as long as injectivity of $\Psi$ is imposed---the image of the causal field theory $\Psi$ is a \emph{causal category} \cite{coecke2013causal,coecke2016terminality}
\footnote{We have taken the liberty to extend the definition of environment structures to partially monoidal categories, such as the image of a $\Psi$ injective on objects under the partial monoidal product induced by the partial monoidal product of the domain category $\mathcal{C}$.}.
Physically, this means that the field evolution happens in a no-signalling way: if the effects $(\trace{\Sigma})_{\Sigma \in \obj{\mathcal{C}}}$ are used as \emph{discarding maps}---generalising the partial traces of quantum theory---then the field state over a given slice $\Sigma$ does not depend on the field state over slices which are in the future of $\Sigma$ or are space-like separated from $\Sigma$.

This emergence of causality and no-signalling from functoriality is in fact a consequence of a breaking of time symmetry which happened in the very definition of the ordering between slices.
Indeed, consider the ``time-reversed'' causal order $\Omega^{rev}$, obtained by reversing all causal relations in $\Omega$ (i.e. $y \leq x$ in $\Omega^{rev}$ if and only if $x \leq y$ in $\Omega$).
The slices for $\Omega^{rev}$ are exactly the slices for $\Omega$, i.e. the categories of all slices $\sliceCat{\Omega^{rev}}$ and $\sliceCat{\Omega}$ have the same objects.
If time symmetry were to hold, we would expect the arrows in $\sliceCat{\Omega^{rev}}$ to be exactly the reverse of the arrows in $\sliceCat{\Omega}$. However, the conditions defining the arrows in both categories are as follows:
\begin{itemize}
    \item $\Sigma \slicePreceq \Gamma$ in $\sliceCat{\Omega}$ iff $\Gamma \subseteq \futuredom{\Sigma}$ in $\Omega$;
    \item $\Gamma \slicePreceq \Sigma$ in $\sliceCat{\Omega^{rev}}$ iff $\Sigma \subseteq \futuredom{\Gamma}$ in $\Omega^{rev}$, i.e. iff $\Sigma \subseteq \pastdom{\Gamma}$ in $\Omega$.
\end{itemize}
The two conditions that $\Gamma \subseteq \futuredom{\Sigma}$ and $\Sigma \subseteq \pastdom{\Gamma}$, both in $\Omega$, are not in general equivalent: this shows that time symmetry is broken by our definition of the relationship between slices, ultimately leading to the emergence of causality and no-signalling constraints on functorial evolution of quantum fields.

\section{Connection with Algebraic Quantum Field Theory}
\label{section:connection-aqft}

The definition of causal field theories looks somewhat similar to that of \emph{Topological Quantum Field Theories (TQFTs)} as functors from categories of cobordism to categories of vector spaces.
The big difference between the causal field theories we defined above and TQFTs is that the latter take the basic building blocks for field theories to be defined over \emph{arbitrary} topological spacetimes, while the former define the evolution over a \emph{single given} spacetime.
This difference is an aspect of a general abstract duality between \emph{compositionality} and \emph{decompositionality}.

In \emph{compositionality}, larger objects are created by \emph{composing} together given elementary building blocks in all possible ways: this is the approach behind an ever growing zoo of process theories (e.g. see \cite{coecke2017picturing} and references therein).
In \emph{decompositionality}, on the other hand, larger objects are given as a whole and subsequently decomposed into smaller constituents, with composition of the latter constrained by the context in which they live: this approach, based on partially monoidal structure, was recently introduced by \cite{gogioso2019process} as a way to talk about compositionality in physical theories where a universe is fixed beforehand.
While TQFTs are compositional \cite{kock2003frobenius,lurie2009classification}, causal field theories are more naturally understood from the decompositional perspective.

In fact, decompositionality is the key ingredient in a completely different family of approaches to quantum theory, including Algebraic Quantum Field Theory (AQFT) \cite{halvorson2006algebraic} and the topos-theoretic approaches \cite{heunen2009topos,doring2008thing}.
In AQFT, specifically, the relationship between fields and the topology of spacetime is encapsulated into the structure of a presheaf, having as its domain the poset formed by diamonds in Minkowski space under inclusion and as its codomain a category of C*-algebras and *-homomorphisms.
To understand the decompositional character of causal field theories, we draw inspiration from the AQFT approach and turn our functors, defined on slices, into presheafs defined on regions.
We will, however, dispense of the algebras themselves: as mentioned earlier in Subsection~\ref{subsection:categories-for-quantum-fields}, our approach is independent of the specific process theory chosen for the fields.

We begin by showing that categories of slices can be restricted to regions, as long as we take care to define regions in such a way as to respect the restrictions imposed by a specific choice of category of slices.

\begin{definition}\label{definition:bounded-regions-categories-of-slices}
    A \emph{bounded region} in a category of slices $\mathcal{C}$ on a causal order $\Omega$ is a region on $\Omega$ in the form $\Diamond_{\Sigma, \Gamma}$ for some $\Sigma, \Gamma \in \obj{\mathcal{C}}$.
    Bounded regions in $\mathcal{C}$ form a poset $\boundedRegionCat{\mathcal{C}}$ under inclusion.
\end{definition}
\begin{definition}\label{definition:regions-categories-of-slices}
    A \emph{region} in a category of slices $\mathcal{C}$ is a region $R$ on $\Omega$ which can be obtained as a union $R = \bigcup_{\lambda \in \Lambda} \Diamond_{\Sigma_\lambda, \Gamma_\lambda}$ of a family $(\Diamond_{\Sigma_\lambda, \Gamma_\lambda})_{\lambda \in \Lambda}$, closed under finite unions, of bounded regions in $\mathcal{C}$.
    Regions in $\mathcal{C}$ also form a poset $\regionCat{\mathcal{C}}$ under inclusion, with $\boundedRegionCat{\mathcal{C}}$ as a sub-poset.
\end{definition}

\noindent Note that if $\mathcal{C} = \sliceCat{\Omega}$ then the regions in $\mathcal{C}$ are exactly the regions on $\Omega$: by definition, a region $R$ on $\Omega$ is the union of the bounded regions $\Diamond_{x,y}$ for all $x,y \in R$.

\begin{proposition}\label{proposition:categories-slices-restriction}
    Let $\mathcal{C}$ be a category of slices and $R$ be a region in it.
    The \emph{restriction} $\mathcal{C}\vert_{R}$ of $\mathcal{C}$ to the region $R$, defined as the full sub-category of $\mathcal{C}$ spanned by the slices $\Delta \in \obj{\mathcal{C}}$ such that $\Delta \subseteq R$, is itself a category of slices.
\end{proposition}
\begin{proof}
    If $R = \Diamond_{\Sigma, \Gamma}$ is a \emph{bounded} region in $\mathcal{C}$, then the statement is an immediate consequence of requirement (2) for categories of slices.
    Now assume that $R = \bigcup_{\lambda \in \Lambda} \Diamond_{\Sigma_\lambda, \Gamma_\lambda}$ is a union of bounded regions in $\mathcal{C}$.

    If $x \leq y$ are two events in $R$, then it must be that $x \in \Diamond_{\Sigma_{\lambda_x}, \Gamma_{\lambda_x}}$ and $y \in \Diamond_{\Sigma_{\lambda_y}, \Gamma_{\lambda_y}}$ for some $\lambda_x, \lambda_y \in \Lambda$: closure under union of the family $(\Diamond_{\Sigma_\lambda, \Gamma_\lambda})_{\lambda \in \Lambda}$ then guarantees that there exists some $\lambda_{x,y} \in \Lambda$ with $x, y \in \Diamond_{\Sigma_{\lambda_{x,y}}, \Gamma_{\lambda_{x,y}}}$.
    Because $\mathcal{C}$ is a category of slices, we can find two slices $\Delta_x \slicePreceq \Delta_y$ in $\mathcal{C}$ such that $x \in \Delta_x$ and $y \in \Delta_y$.
    Then the restrictions $(\Delta_x \cap \Diamond_{\Sigma_{\lambda_{x,y}}, \Gamma_{\lambda_{x,y}}})$ and $(\Delta_y \cap \Diamond_{\Sigma_{\lambda_{x,y}}, \Gamma_{\lambda_{x,y}}})$ satisfy requirement (1) for $\mathcal{C}\vert_{R}$ to be a category of slices.

    If $\Sigma$, $\Gamma$ and $\Delta$ are three slices in $\mathcal{C}\vert_{R}$, then in particular the diamond $\Diamond_{\Sigma, \Gamma}$ is a subset of $R$ (the latter is a region) and so is the intersection $\Delta \cap \Diamond_{\Sigma, \Gamma}$, which exists in $\mathcal{C}$ because the latter is a category of slices.
    Hence requirement (2) for $\mathcal{C}\vert_{R}$ to be a category of slices is satisfied.

    Requirement (3) for $\mathcal{C}\vert_{R}$ to be a category of slices is satisfied, because if $\Sigma, \Gamma \subseteq R$ then also $\Sigma \otimes \Gamma \subseteq R$ whenever the latter is defined.
\hfill$\square$\end{proof}

Given a causal field theory $\Psi: \mathcal{C} \rightarrow \mathcal{D}$, the restrictions $\Psi\vert_{R}: \mathcal{C}\vert_{R} \rightarrow \mathcal{D}$ are again causal field theories.
To match the spirit of AQFT, we need two more ingredients: the definition of a space of states $\states{\Psi}{R}$ over a region $R$ and the definition of restrictions $\states{\Psi}{R} \rightarrow \states{\Psi}{R'}$ between spaces of states associated with inclusions $R' \subseteq R$ of regions.

\begin{definition}\label{definition:states-on-regions}
    Given a region $R$ in a category of slices $\mathcal{C}$, the \emph{space of states} $\states{\Psi}{R}$ over the region is defined to be the set comprising all families $\rho$ of states over the slices in $\mathcal{C}\vert_{R}$ which are stable under the action of $\Psi$, i.e. comprising all the families
    \begin{equation}
        \rho
        \in \hspace{-5mm}
        \prod_{
            \Delta \in \obj{\mathcal{C}\vert_{R}}
        } \hspace{-5mm} \states{\mathcal{D}}{\Psi(\Delta)}
    \end{equation}
    such that for all $\Delta, \Delta' \in \obj{\mathcal{C}\vert_{R}}$ with $\Delta \slicePreceq \Delta'$ the following condition is satisfied:
    \begin{equation}
        \Psi(\Delta \slicePreceq \Delta') \circ \rho_{\Delta} = \rho_{\Delta'}
    \end{equation}
    By $\states{\mathcal{D}}{\Psi(\Delta)}$ we have denoted the states on the object $\Psi(\Delta)$ of the symmetric monoidal category $\mathcal{D}$, i.e. the homset $\Hom{\mathcal{D}}{I}{\Psi(\Delta)}$ where $I$ is the monoidal unit of $\mathcal{D}$.
\end{definition}
\begin{proposition}\label{proposition:presheaf-of-states-on-regions}
    Given a causal field theory $\Psi: \mathcal{C} \rightarrow \mathcal{D}$, we can construct a presheaf $\statesPresheaf{\Psi}: \regionCat{\mathcal{C}}^{op} \rightarrow \setCat$ by associating each region $R \in \obj{\regionCat{\mathcal{C}}}$ to the space of states $\states{\Psi}{R}$ over the region, and each inclusion $i: R' \subseteq R$ to the restriction function $\states{\Psi}{R} \rightarrow \states{\Psi}{R'}$ defined by sending a family $\rho \in \states{\Psi}{R}$ to the family $\statesPresheaf{\Psi}(i)(\rho) \in \states{\Psi}{R'}$ given as follows:
    \begin{equation}
        \statesPresheaf{\Psi}(i)(\rho)_{\Delta'} = \rho_{i(\Delta')}
    \end{equation}
    We refer to $\statesPresheaf{\Psi}$ as the \emph{presheaf of states} over regions of $\,\mathcal{C}$.
\end{proposition}
\begin{proof}
    The only thing to show is functoriality of $\statesPresheaf{\Psi}$.
    If $i = \id{R}: R \subseteq R$ is the identity on a region $R$, then we have:
    \begin{equation}
        \statesPresheaf{\Psi}(i)(\rho)_{\Delta} = \rho_{i(\Delta)} = \rho_{\Delta}
    \end{equation}
    i.e. $\statesPresheaf{\Psi}(i) = \id{\states{\Psi}{R}}$ is the identity on the space of states over the region.
    If now $j: R'' \subseteq R'$ and $i: R' \subseteq R$, then $i \circ j: R'' \subseteq R$ and we have:
    \begin{equation}
        \statesPresheaf{\Psi}(j)\left(\statesPresheaf{\Psi}(i)(\rho)\right)_{\Delta''}
        =
        \statesPresheaf{\Psi}(i)(\rho)_{j(\Delta'')}
        =
        \rho_{i(j(\Delta''))}
        =
        \statesPresheaf{\Psi}(i \circ j)(\rho)_{\Delta''}
    \end{equation}
    Hence $\statesPresheaf{\Psi}$ is a presheaf $\statesPresheaf{\Psi}: \regionCat{\mathcal{C}}^{op} \rightarrow \setCat$.
\hfill$\square$\end{proof}

\begin{definition}\label{definition:global-state}
    A \emph{global state} $\rho$ for a causal field theory $\Psi: \mathcal{C} \rightarrow \mathcal{D}$ is a global compatible family for $\statesPresheaf{\Psi}$, i.e. a family $\rho = \left(\rho^{(R)}\right)_{R \in \regionCat{\mathcal{C}}}$ such that $\statesPresheaf{\Psi}(i)(\rho^{(R)}) = \rho^{(R')}$ for all inclusions $i: R' \subseteq R$ in $\regionCat{\mathcal{C}}$.
    We refer to the set of all global states as the \emph{space of global states}.
\end{definition}

\begin{remark}\label{remark:global-states-notation}
    If $\Omega$ is a region in $\mathcal{C}$, i.e. if $\Omega \in \regionCat{\mathcal{C}}$, then the states in $\states{\Psi}{\Omega}$ ($\Omega$ as a region) are in bijection with the global states as follows:
    \begin{equation}
        \begin{cases}
        \rho \in \states{\Psi}{\Omega}\mapsto \left(\statesPresheaf{\Psi}(R \subseteq \Omega)(\rho)\right)_{R \in \regionCat{\mathcal{C}}} \\
        (\rho^{(R)})_{R \in \regionCat{\mathcal{C}}} \text{ \textnormal{global state}} \mapsto \rho^{(\Omega)}
        \end{cases}
    \end{equation}
    Because of this, we consistently adopt the notation $\states{\Psi}{\Omega}$ to denote the space of global states. If $R$ is a region in $\mathcal{C}$, we also adopt the notation $\statesPresheaf{\Psi}(R \subseteq \Omega)$ for the map $\states{\Psi}{\Omega} \rightarrow \states{\Psi}{R}$ sending a global state $\rho = (\rho^{(R)})_{R \in \regionCat{\mathcal{C}}}$ to its component $\rho^{(R)}$ over the region $R$.
    To unify notation, we will also adopt $\rho_\Sigma$ to denote $(\rho^{(R)})_\Sigma$, taking the same value for any region $R$ containing $\Sigma$ (e.g. for $R = \Sigma$).
\end{remark}

\begin{remark}\label{remark:enrichment}
    If the field category $\mathcal{D}$ is suitably enriched (e.g. in a category with all limits), then a natural choice is for the the space of states to be defined by a presheaf valued in the enrichment category.
    For example, quantum theory is enriched over positive cones, i.e. $\reals^+$-modules, and the $\reals^+$-linear structure of states in quantum theory extends to a $\reals^+$-linear structure on the spaces of states of causal field theories having quantum theory as their field category.
    We will not consider such enrichment in this work, though all constructions we present can be readily extended to such a setting.
\end{remark}

Spaces of states according to Definition~\ref{definition:states-on-regions} encode a lot of redundant information, because we don't want to look into the specific structure of regions. However, there are certain special cases in which an equivalent description of the space of states over a region can be given.

To start with, consider consider two slices $\Sigma \slicePreceq \Gamma$ and note that the state on any slice $\Delta \subseteq \Diamond_{\Sigma, \Gamma}$ in a bounded region $\Diamond_{\Sigma, \Gamma}$ is uniquely determined by applying $\Psi(\Sigma \slicePreceq \Delta)$ to the state on $\Sigma$:
\begin{equation}
    \rho_\Delta = \Psi(\Sigma \slicePreceq \Delta)(\rho_\Sigma)
\end{equation}
This is, for example, the case for all bounded regions between Cauchy slices in a category of slices $\cauchySliceCat{\mathcal{F}}$ generated by some foliation $\mathcal{F}$.
If the foliation $\mathcal{F}$ has a minimum $\Sigma_{0}$---an \emph{initial Cauchy slice}---then any global state $\rho \in \states{\Psi}{\Omega}$ is entirely determined by its component $\rho_{\Sigma_{0}}$ over the initial slice $\Sigma_0$:
\begin{equation}
    \rho_\Delta = \Psi(\Sigma_{0} \slicePreceq \Delta) \circ \rho_{\Sigma_{0}}
\end{equation}
for any $\Delta \in \mathcal{F}$ and any region $R$ in $\cauchySliceCat{\mathcal{F}}$ such that $\Delta \subseteq R$. This extends to all slices in $\cauchySliceCat{\mathcal{F}}$ by restriction.

Inspired by Relativity, we would like the state on \emph{any} Cauchy slice in the foliation to determine the global state, not only that on an initial Cauchy slice (which may not exist).
For this to happen, we need to strengthen our requirements on the causal field theory, which needs to be \emph{causally reversible}.

\begin{definition}\label{definition:reverse-causal-order}
    Let $\Omega$ be any causal order. By the \emph{causal reverse} of $\,\Omega$ we mean the causal order $\Omega^{rev}$ on the same events as $\Omega$ and such that $x \leq y$ in $\Omega^{rev}$ if and only if $x \geq y$ in $\Omega$.
\end{definition}

\begin{definition}\label{definition:reversible-category-of-slices}
    A category of slices $\mathcal{C}$ on a causal order $\Omega$ is said to be \emph{causally reversible} if the full sub-category of $\sliceCat{\Omega^{rev}}$ spanned by $\obj{\mathcal{C}}$ is a category of slices on the causal reverse $\Omega^{rev}$.
    If this is the case, we write $\mathcal{C}^{rev}$ for said category of slices over $\Omega^{rev}$ and refer to it as the \emph{causal reverse} of $\mathcal{C}$.
    We write $\stackrel{rev}{\slicePreceq}$ for the morphisms of $\mathcal{C}^{rev}$.
\end{definition}

\begin{definition}\label{definition:reversible-cft}
    Let $\Psi: \mathcal{C} \rightarrow \mathcal{D}$ be a causal field theory on a causal order $\Omega$. If $\mathcal{C}$ is causally reversible, a \emph{causal reversal} of $\Psi$ is a causal field theory $\Phi: \mathcal{C}^{rev} \rightarrow \mathcal{D}$ such that:
    \begin{enumerate}
        \item[(1)] the functors $\Psi$ and $\Phi$ agree on objects, i.e. for all $\Sigma \in \obj{\mathcal{C}}$ we have that $\Psi(\Sigma) = \Phi(\Sigma)$;
        \item[(2)] whenever we have two chains of alternating morphisms in $\mathcal{C}$ and $\mathcal{C}^{rev}$ which start and end at the same slices $\Sigma, \Gamma$, say in the form
        \begin{align}
            \Sigma \slicePreceq \Delta_1 \stackrel{rev}{\slicePreceq} \Delta_2 \slicePreceq ... \Delta_{2n} \slicePreceq \Gamma \nonumber\\
            \Sigma \slicePreceq \Delta'_1 \stackrel{rev}{\slicePreceq} \Delta'_2 \slicePreceq ... \Delta'_{2m} \slicePreceq \Gamma
        \end{align}
        for some $n,m \geq 0$, the composition of the images of the morphisms under $\Psi$ and $\Phi$ always yield the same morphism $\Psi(\Sigma) \rightarrow \Psi(\Gamma)$:
        \begin{align}
            &\Psi(\Delta_{2n} \slicePreceq \Gamma)
            \circ
            ...
            \circ
            \Phi(\Delta_1 \stackrel{rev}{\slicePreceq} \Delta_2)
            \circ
            \Psi(\Sigma \slicePreceq \Delta_1)\nonumber\\
            =
            &\Psi(\Delta'_{2m} \slicePreceq \Gamma)
            \circ
            ...
            \circ
            \Phi(\Delta'_1 \stackrel{rev}{\slicePreceq} \Delta'_2)
            \circ
            \Psi(\Sigma \slicePreceq \Delta'_1)
        \end{align}
    \end{enumerate}
    We say that $\Psi: \mathcal{C} \rightarrow \mathcal{D}$ is \emph{causally reversible}---or simply \emph{reversible}---if $\mathcal{C}$ is causally reversible and $\Psi$ admits a causal reversal.
\end{definition}

\begin{proposition}\label{proposition:foliations-causally-reversible-cft}
    Let $\cauchySliceCat{\mathcal{F}}$ be the category of slices on a causal order $\Omega$ generated by some foliation $\mathcal{F}$.
    Then $\cauchySliceCat{\mathcal{F}}$ is always causally reversible and for any two Cauchy slices $\Delta, \Sigma$ we have that $\Delta \slicePreceq \Sigma$ if and only if $\Sigma \stackrel{rev}{\slicePreceq} \Delta$.
    Furthermore, if a causal field theory $\Psi: \cauchySliceCat{\mathcal{F}} \rightarrow \mathcal{D}$ is reversible, then a global state $\rho$ is entirely determined by the state $\rho_\Sigma$ on \emph{any} Cauchy slice $\Sigma \in \mathcal{F}$ as follows:
    \begin{equation}
        \label{equation:foliations-causally-reversible-cft}
        \rho_\Delta =
        \begin{cases}
            \Psi(\Sigma \slicePreceq \Delta) \circ \rho_\Sigma & \text{ if } \Sigma \slicePreceq \Delta\\
            \Phi(\Sigma \stackrel{rev}{\slicePreceq} \Delta) \circ \rho_\Sigma & \text{ if } \Delta \slicePreceq \Sigma
        \end{cases}
    \end{equation}
    where $\Phi: \cauchySliceCat{\mathcal{F}}^{rev} \rightarrow \mathcal{D}$ is any causal reversal of $\Psi$.
\end{proposition}
\begin{proof}
    The main observation behind this result is as follows: if $\Sigma, \Delta$ are two Cauchy slices, then the conditions $\Delta \subseteq \futuredom{\Sigma}$ and $\Sigma \subseteq \pastdom{\Delta}$ are equivalent.
    Hence $\cauchySliceCat{\mathcal{F}}$ is always causally reversible and $\Delta \slicePreceq \Sigma$ if and only if $\Sigma \stackrel{rev}{\slicePreceq} \Delta$ for any two Cauchy slices $\Delta, \Sigma$.

    Now let $\Psi$ be causally reversible, let $\Sigma \in \mathcal{F}$ be a Cauchy slice in the foliation and consider any global state $\rho$.
    If $\Sigma \slicePreceq \Delta$ for some other Cauchy slice $\Delta \in \mathcal{F}$, then the definition of a global state implies that $\rho_\Delta = \Psi(\Sigma \slicePreceq \Delta) \circ \rho_\Sigma$.
    If instead $\Delta \slicePreceq \Sigma$, then $\Sigma \stackrel{rev}{\slicePreceq} \Delta$ and the definition of a global state implies that $\rho_\Sigma = \Psi(\Delta \slicePreceq \Sigma) \circ \rho_\Delta$. But the definition of a causal reverse also implies that:
    \begin{equation}
        \Phi(\Sigma \stackrel{rev}{\slicePreceq} \Delta) \circ \rho_\Sigma
        =
        \Phi(\Sigma \stackrel{rev}{\slicePreceq} \Delta) \circ \Psi(\Delta \slicePreceq \Sigma) \circ \rho_\Delta
        =
        \Psi(\Sigma \slicePreceq \Sigma)  \circ \rho_\Delta
        =
        \id{\Psi(\Sigma)}  \circ \rho_\Delta
        =
        \rho_\Delta
    \end{equation}
    Hence the value $\rho_\Sigma$ completely determines the global state $\rho$ (since the value on all other slices in $\cauchySliceCat{\mathcal{F}}$ is determined by restriction from the value on a corresponding Cauchy slice).
\hfill$\square$\end{proof}

It is an easy check that not only the global states $\rho \in \states{\Psi}{\Omega}$ are determined---under the conditions of Proposition~\ref{proposition:foliations-causally-reversible-cft}---by their component $\rho_\Sigma \in \states{\mathcal{D}}{\Psi(\Sigma)}$ over any Cauchy slice $\Sigma$ in the foliation, but also that Equation~\ref{equation:foliations-causally-reversible-cft} can be used---under the same conditions---to construct a global state $\rho \in \states{\Psi}{\Omega}$ from a state $\rho_\Sigma \in \states{\mathcal{D}}{\Psi(\Sigma)}$ on any Cauchy slice $\Sigma$ in the foliation.

Before concluding this Section, we would like to remark that a succinct description of spaces of states over regions can be obtained in settings much more general than those of foliations: for example, in all those cases where the every region admits a suitable Cauchy slice and the causal field theory is reversible.
The careful formulation of this more general setting is key to the further development of the connection between causal field theory and AQFT and it is left to future work.


\section{Connection to quantum cellular automata}
\label{section:connection-qca}

The idea of a cellular automaton was first introduced by von Neumann, aimed at designing a self replicating machine \cite{vonNeumann1966automata}. A \emph{Cellular Automaton} (CA) over some finite alphabet $A$ has its state stored as a $d$-dimensional lattice of values in $A$, i.e. as a function $\psi: \integers^d \rightarrow A$. The state is updated at discrete time steps, each step updated as $\psi^{(t+1)} := F(\psi^{(t)})$ according to some fixed function $F: (\integers^d \rightarrow A) \rightarrow (\integers^d \rightarrow A)$.
The function $F$ acts \emph{locally} and \emph{homogeneously}: there is some fixed finite subset $\mathcal{N} \subset \integers^d$ (typically a neighbourhood of $\underline{0} \in \integers^d$) and some function $f: \mathcal{N} \rightarrow A$ such that the value of each lattice site $\underline{x}$ at time step $t+1$ only depends on the finitely many values in the subset $\underline{x} + \mathcal{N}$ at time $t$:

\begin{equation}
    F(\psi) := \underline{x} \mapsto f(\psi\vert_{\underline{x} + \mathcal{N}})
\end{equation}

A \emph{Quantum Cellular Automaton} (QCA) is a generalization of a CA where the lattice states $\psi: \integers^d \rightarrow A$ are replaces by (pure) states in the tensor product of Hilbert spaces $\bigotimes_{\underline{x} \in \integers^d} \mathcal{H}_x$ (all $\mathcal{H}_{\underline{x}}$ finite-dimensional and isomorphic) and the function $F$ is replaced by a unitary $U: \bigotimes_{\underline{x} \in \integers^d} \mathcal{H}_{\underline{x}} \rightarrow \bigotimes_{\underline{x} \in \integers^d} \mathcal{H}_{\underline{x}}$, with requirements of locality and homogeneity.

\begin{remark}\label{remark:infinite-tensor-product}
    There are several slightly different formulation of the infinite tensor product above that can be used, each with its own advantages and disadvantages: though it is not going to be  a concern for this work, the authors are partial to the construction by von Neumann \cite{vonNeumann1939tensor}.
\end{remark}

An early formulation of the notion of QCA is due to Richard Feynman, in the context of simulations of physics using quantum computers \cite{feynmann1982simulating}. More recent work on quantum information and quantum causality has shown that the evolution of certain free quantum fields can be recovered as the continuous limit of certain quantum cellular automata (cf. \cite{dariano2016automata,arrighi2019automata} and references therein).
In the final section of this work, we show that our framework is well-suited to capture notions of QCA such as those appearing in the literature. Specifically, our construction encompasses and greatly generalises that presented in \cite{arrighi2019automata}.

\subsection{Causal cellular automata}

The first requirement in the definition of a QCA is that of \emph{homogeneity}---called ``translation invariance'' in \cite{arrighi2019automata}---i.e. the requirement that the automaton act the same way at all points of spacetime.
Because presentations of QCAs are usually given in terms of discrete updates of states on a lattice by means of a unitary $U$, only the requirement of homogeneity \emph{in space} is usually mentioned.
However, such presentations also have homogeneity in time as an implicit requirement, namely in the assumption that the same unitary $U$ be used to update the state at all times.

Instead of updating the state time-step by time-step in a compositional fashion, our formulation of quantum cellular automata will see the entirety of spacetime at once, with states over slices and regions recovered in a decompositional approach.
Nevertheless, the requirement of homogeneity for a QCA can still be formulated as a requirement of invariance under certain symmetries of spacetime, so we begin by formulating such a notion of invariance for causal field theories.

\begin{definition}\label{definition:causal-order-symmetries}
    A \emph{symmetry} on a causal order $\Omega$ is an action of a group $G$ on $\Omega$ by automorphisms of causal orders, i.e. a group homomorphism $G \rightarrow \Automs{\causOrdCat}{\Omega}$.
    If $\mathcal{C}$ is a category of slices on $\Omega$, a \emph{symmetry} on $\mathcal{C}$ is a symmetry on $\Omega$ which extends to an action on $\mathcal{C}$ by partially monoidal functors, i.e. one such that the following conditions are satisfied:
    \begin{itemize}
        \item[(1)] for all $g \in G$, if $\Sigma \in \obj{\mathcal{C}}$ then $g(\Sigma) \in \obj{\mathcal{C}}$;
        \item[(2)] for all $g \in G$ and all $\Sigma, \Gamma \in \obj{\mathcal{C}}$, if $\Sigma \slicePreceq \Gamma$ then $g(\Sigma) \slicePreceq g(\Gamma)$;
        \item[(3)] for all $g \in G$ and all $\Sigma, \Gamma \in \obj{\mathcal{C}}$, if $\Sigma \otimes \Gamma$ is defined in $\mathcal{C}$ then $g(\Sigma \otimes \Gamma) = g(\Sigma) \otimes g(\Gamma)$ is also defined in $\mathcal{C}$.
    \end{itemize}
    Note, for all $g \in G$, that $g(\emptyset) = \emptyset$ and that $g(\Sigma)$ is automatically a slice whenever $\Sigma$ is a slice.
\end{definition}

\begin{definition}\label{definition:symmetry-invariant-cft}
    Let $\mathcal{C}$ is a category of slices with a symmetry action of a group $G$.
    A \emph{$G$-invariant} (or simply \emph{symmetry-invariant}) causal field theory on $\mathcal{C}$ is a causal field theory $\Psi: \mathcal{C} \rightarrow \mathcal{D}$ equipped with a family of natural isomorphisms $\Psi \stackrel{\alpha_g}{\Rightarrow} \Psi \circ g$ such that $\alpha_{h \cdot g} = \alpha_hg \circ \alpha_g$, where we have again identified elements $g \in G$ with their action as partially monoidal functors $g: \mathcal{C} \rightarrow \mathcal{C}$.
\end{definition}

\begin{remark}\label{remark:symmetry-invariant-cft}
    The spirit behind the definition of symmetry-invariant causal field theories is that the functors $\Psi$ (sending slices $\mapsto$ fields) and $\Psi \circ g$ (sending slices $\mapsto$ $g$-translated slices $\mapsto$ fields) should be the same. However, we have remarked when first defining causal field theories that---be it for ease of physical interpretation or for conformity with existing literature on causal categories---it may sometimes be desirable that the images $\Psi(\Sigma)$ of different slices be different.
    Not being able to impose the equality $\Psi = \Psi \circ g$ in such a setting, the next best thing is to ask for natural isomorphism $\Psi \cong \Psi \circ g$.

    Because we are dealing with symmetries, however, it is sensible to require for the natural isomorphisms $\alpha_g$ themselves to respect the group structure.
    Again the first instinct might be to require something in the form $\alpha_{h \cdot g} = \alpha_{h} \circ \alpha_g$, but this expressions does not type-check: we have a natural transformation $\alpha_{h \cdot g}: \Psi \Rightarrow \Psi \circ h \circ g$, a natural transformation $\alpha_g: \Psi \Rightarrow \Psi \circ g$ and a natural transformation $\alpha_h: \Psi \Rightarrow \Psi \circ h$. In order to compose $\alpha_h$ and $\alpha_g$ we instead have to take the action of $\alpha_h$ \emph{translated} to $\Psi \circ g$:
    \begin{equation}
        \alpha_hg: \Psi \circ g \Rightarrow (\Psi \circ h) \circ g
    \end{equation}
    Explicitly, the natural transformation $\alpha_hg$ is defined by $(\alpha_hg)(\Sigma) := \alpha_h(g(\Sigma))$.
\end{remark}

The second requirement in the definition of a QCA is that of ``locality''.
When quantum cellular automata are considered in a relativistic context---e.g. as discrete models of quantum field theories---the requirement of locality is called \emph{causality}, as it is meant to capture the idea that the action of the automaton should respect the causal structure of spacetime (so that the state on a point $\underline{x}$ at time $t + \Delta t$ should not depend on the state at the previous time $t$ on points $\underline{y}$ which are ``too far away'', i.e. such that $(\underline{x}, t + \Delta t)$ and $(\underline{y}, t)$ are space-like separated).

In \cite{arrighi2019automata}, causality is formulated as the requirement that the output state of the automaton over a point $\underline{x}$ of the lattice at time $t+1$ only depend on the state over a finite neighbourhood $\underline{x} + \mathcal{N}$ at time $t$.
In our framework, on the other hand, causality is automatically enforced: the state over a slice never depends on the state on any other slice which is space-like separated from it.

\begin{remark}
    The causal order $\Omega$ which captures the causality requirement from \cite{arrighi2019automata} with finite neighbourhood $\mathcal{N} \subset \integers^d$ can be constructed by endowing the set $|\Omega| := \integers^d \times \integers$ with the reflexive-transitive closure of the relation $(\underline{y}, t) \leq (\underline{x}, t+1)$ for all times $t \in \integers$, for all points of the lattice $\underline{x} \in \integers^d$ and for all points $\underline{y} \in \underline{x} + \mathcal{N}$ in the neighbourhood of $\underline{x}$.
\end{remark}

The third and final requirement in the definition of a QCA is that of \emph{unitarity}.
In our framework, this is a problem for two (mostly unrelated) reasons.
\begin{itemize}
    \item Our formulation of causal field theories aims to be agnostic to the choice of process theory. On the other hand, unitarity is a strongly quantum-like feature, the formulation of which would require a significant amount of additional structure on the field category.
    \item The usual formulation of quantum cellular automata only considers global evolution, never directly dealing with restrictions---situations e.g. in which the state is evolved unitarily but part of the output state is discarded as environment. Our framework instead treats such restrictions as an integral part of evolution.
\end{itemize}
Luckily, unitarity \emph{per se} is not necessary from an abstract foundational standpoint: the real feature of interest is \emph{reversibility}, a feature of causal field theories which we have already explored.
For the sake of generality, we will not include reversibility in the definition below, leaving it as an explicit desideratum.

\begin{definition}\label{definition:qca}
    A \emph{Causal Cellular Automaton} (CCA) consists of the following ingredients.
    \begin{enumerate}
        \item[(1)] A foliation $\mathcal{F}$ on a causal order $\Omega$.
        \item[(2)] A category of Cauchy slices $\mathcal{C}$ such that each slice in $\mathcal{C}$ is a subset of some Cauchy slice in $\mathcal{F}$.
        \footnote{Each Cauchy slice $\Sigma$ in $\mathcal{F}$ is then automatically the union of all slices $\Delta \in \obj{\mathcal{C}}$ such that $\Delta \subseteq \Sigma$.}
        \item[(3)] A symmetry action of a group $G$ on $\mathcal{C}$, inducing---via the $G$-action on $\Omega$---a transitive action of $G$ on the Cauchy slices in the foliation $\mathcal{F}$.
        \item[(4)] A $G$-invariant causal field theory $\Psi: \cauchySliceCat{\mathcal{F}} \rightarrow \mathcal{D}$.
    \end{enumerate}
    A \emph{reversible} CCA is one where the causal field theory $\Psi$ is reversible.
\end{definition}

Definition~\ref{definition:qca} is much more general than the usual definition of a QCA and hence captures more sophisticated examples. However, its ingredients are directly analogous to those appearing in the definition of a QCA.
\begin{itemize}
    \item The foliation $\mathcal{F}$ on $\Omega$ generalises the discrete time steps in the definition of a QCA.
    \item The slices in $\mathcal{C}$ generalise the equal-time hyper-surfaces which support the state of a QCA at fixed time.
    \item The symmetry action of $G$ on $\cauchySliceCat{\mathcal{F}}$ and its transitivity on the foliation $\mathcal{F}$ generalise homogeneity in both space and in time of the lattices supporting a QCA.
    \item The $G$-invariance of the causal field theory $\Psi$ generalises both the translation symmetry in space and the time-translation symmetry of a QCA.
\end{itemize}

\subsection{Partitioned causal cellular automata}

We now proceed to construct a large family of examples of CCAs based on the \emph{partitioned QCAs} of \cite{arrighi2019automata}. In doing so, we generalise the scattering unitaries to arbitrary processes and allow for the definition of state restriction to non-Cauchy equal-time surfaces. We refer to the resulting CCA as \emph{partitioned CCA}.

\subsubsection{Causal order}
As our causal order $\Omega$ we consider the following subset of $(1+d)$-dimensional Minkowski spacetime (setting the constant $c$ for the speed of light to $c = \sqrt{d}$):
\begin{equation}
    \Omega := \suchthat{(t, \underline{x})}{t \in \integers, \underline{x} \in (t,...,t) + 2\integers^{d}}
\end{equation}
where $(t,...,t) + 2\integers^{d}$ is the set of all $\underline{x} \in \integers^{d}$ such that $x_i = \modclass{t}{2}$.
For $d = 1$ we get the $(1+1)$-dimensional diamond lattice discussed before. In general, the immediate causal predecessors of a point $(t, \underline{x})$ are the following $2^d$ points:
\begin{equation}
    \left(t-1, \underline{x} - \mathcal{N}\right) = \suchthat{(t-1, \underline{x} - \underline{\delta})}{\underline{\delta} \in \mathcal{N}}
\end{equation}
where we defined the ``neighbourhood'' $\mathcal{N} := \{\pm1\}^d$.
Similarly, the immediate successors of $(t, \underline{x})$ are the following $2^d$ points:
\begin{equation}
    \left(t+1, \underline{x} + \mathcal{N}\right) = \suchthat{(t-1, \underline{x} + \underline{\delta})}{\underline{\delta} \in \mathcal{N}}
\end{equation}

\subsubsection{Foliation and category of slices}

The causal order $\Omega$ admits a foliation $\mathcal{F}$ where each slice is a constant-time Cauchy slice $\Sigma_t$ for some $t \in \integers$:
\begin{equation}
    \Sigma_t := \suchthat{(t, \underline{x})}{\underline{x} \in (t,...,t) + 2\integers^{d}}
\end{equation}
A suitable category of slices $\mathcal{C}$ to associate to this foliation is given by taking as slices all the finite sets $\Sigma_{t, \mathcal{X}} \subset \Sigma_t$ of events having the same time coordinate $t$:
\begin{equation}
    \Sigma_{t, \mathcal{X}} = \suchthat{(t, \underline{x})}{\underline{x} \in \mathcal{X}}
\end{equation}
where  $\mathcal{X} \subset (t,...,t) + 2\integers^{d}$ is some finite subset.
The morphisms $\slicePreceq$ of $\mathcal{C}$ are given as follows for $k \geq 0$:
\begin{equation}
    \label{equation:partitioned-cca-slicepreceq}
    \Sigma_{t, \mathcal{X}} \slicePreceq \Sigma_{t+k, \mathcal{Y}}
    \hspace{5mm} \text{ if and only if } \hspace{5mm}
    \bigcup_{\underline{y} \in \mathcal{Y}} \left(\left(t, \underline{y} + \mathcal{N}^{(k)}\right)\right)
    \subseteq \mathcal{X}
\end{equation}
where the ``iterated neighbourhood'' $\mathcal{N}^{(k)}$ is defined as $\mathcal{N} + ... + \mathcal{N}$ by adding together $k \geq 0$ copies of $\mathcal{N}$ (and we set $\mathcal{N}^{(0)} := \{0\}$). Explicitly we have:
\begin{equation}
    \mathcal{N}^{(k)}
    :=
    \begin{cases}
        \{-k, -k+2, ...-1,+1,...,k-2,k\} & \text{ if $k$ odd}\\
        \{-k, -k+2, ...-2, 0, +2,...,k-2,k\} & \text{ if $k$ even}
    \end{cases}
\end{equation}
It is easy to check (by a $t \mapsto -t$ symmetry argument) that $\mathcal{C}$ is reversible.

\subsubsection{Symmetry}

The category $\mathcal{C}$ admits a symmetry action of the group $G := \integers^{\mathcal{N}} \cong \integers^{2^d}$. We index the coordinates of vectors in $\integers^{\mathcal{N}}$ by the $2^d$ points $\underline{\delta} \in \mathcal{N} = \{\pm1\}^d$. We denote by $\tau_{\underline{\delta}}$ the vector in $\integers^{\mathcal{N}}$ which is $1$ at the coordinate labelled by $\underline{\delta}$ and $0$ at all other coordinates. The action is then specified by setting:
\begin{equation}
    \tau_{\underline{\delta}}(t, \underline{x}) := \left(t+1, \underline{x} -\underline{\delta}\right)
\end{equation}
That is, the $2^d$ generators of $\integers^{\mathcal{N}}$ send a generic event $(t, \underline{x})$ to each of its $2^d$ immediate causal successors in $\Omega$, one for each possible choice of sign $\pm 1$ along each of the $d$ directions of the space lattice $\integers^d$.
\footnote{The reason for the negative sign in $\underline{x} -\underline{\delta}$ is that $\mathcal{N}$ was originally defined to be the neighbourhood \emph{in the past}.}
Each generator $\tau_{\underline{\delta}}$ for the symmetry action sends a Cauchy slice $\Sigma_t$ in the foliation to the next Cauchy slice $\Sigma_{t+1}$, so the action of $G$ on the foliation is transitive.

\subsubsection{Causal field theory - field over slices}

As our field category we consider a generic causal process theory $\mathcal{D}$, i.e. a symmetric monoidal category equipped with a family of \emph{discarding maps} $\trace{\mathcal{H}}: \mathcal{H} \rightarrow I$ for all objects $\mathcal{H} \in \obj{\mathcal{D}}$, respecting the tensor product $\otimes$ and tensor unit $I$ of $\mathcal{D}$: $\trace{\mathcal{H} \otimes \mathcal{K}} = \trace{\mathcal{H}} \otimes \trace{\mathcal{K}}$ and $\trace{I} = 1$.
Discarding maps generalise the partial trace of quantum theory: normalised states $\rho: I \rightarrow \mathcal{H}$---generalising density matrices---are defined to be those such that $\trace{\mathcal{H}} \circ \rho = 1$ and normalised morphisms $U: \mathcal{H} \rightarrow \mathcal{K}$---generalising CPTP maps---are defined to be those such that $\trace{\mathcal{K}} \circ U = \trace{\mathcal{H}}$.
See e.g. \cite{gogioso2017categorical,coecke2013causal,coecke2017picturing} for more information.

To create a $G$-invariant causal field theory $\Psi$, we consider some object $\mathcal{H} \in \obj{\mathcal{D}}$ together with some endomorphism $U: \mathcal{H}^{\otimes 2^d} \rightarrow \mathcal{H}^{\otimes 2^d}$, which we will refer to as the \emph{scattering map}.
For reasons that will soon become clear, it is more convenient to index the factors of $\mathcal{H}^{\otimes 2^d}$ by the $2^d$ points in the neighbourhood $\mathcal{N}$, hence writing $U: \mathcal{H}^{\otimes\mathcal{N}} \rightarrow \mathcal{H}^{\otimes\mathcal{N}}$.

We define the action of $\Psi$ on the slices in $\mathcal{C}$ as follows:
\begin{equation}
    \Psi(\Sigma_{t, \mathcal{X}}) := \left(\mathcal{H}^{\otimes\mathcal{N}}\right)^{\otimes \mathcal{X}} = \mathcal{H}^{\otimes(\mathcal{N} \times \mathcal{X})}
\end{equation}
The tensor product is well-defined in all symmetric monoidal categories, since $\mathcal{X}$ is always finite.
Physically, the field takes values in a copy of $\mathcal{H}^{\otimes\mathcal{N}}$ over each event $(t, \underline{x})$ of spacetime, each individual $\mathcal{H}$ factor of $\mathcal{H}^{\otimes\mathcal{N}}$ encoding the contribution to the field state at $(t, \underline{x})$ from the field state at each of its immediate causal predecessors in $(t-1, \underline{x} + \mathcal{N})$.

\subsubsection{Causal field theory - restriction and evolution}

From their definition in Equation~\ref{equation:partitioned-cca-slicepreceq}, it is easy to see that morphisms $\Sigma_{t, \mathcal{X}_0} \slicePreceq \Sigma_{t+k, \mathcal{X}_k}$ on $\mathcal{C}$ can always be factored in the following way:
\begin{equation}
    \Sigma_{t, \mathcal{X}_0}
    \slicePreceq
    \Sigma_{t, \mathcal{Y}_0}
    \slicePreceq
    \Sigma_{t+1, \mathcal{X}_1}
    \slicePreceq
    \Sigma_{t+1, \mathcal{Y}_1}
    \slicePreceq
    ...
    \slicePreceq
    \Sigma_{t+k, \mathcal{X}_k}
\end{equation}
where $\mathcal{Y}_i \subseteq \mathcal{X}_i$ for all $i = 0,...,k-1$ and the following holds for each $i = 1, ..., k$:
\begin{equation}
    \mathcal{Y}_{i-1}
    =
    \bigcup_{\underline{x} \in \mathcal{X}_i} \suchthat{\left(t+i-1, \underline{x} + \underline{\delta}\right)}{\underline{\delta} \in \mathcal{N}}
\end{equation}
This means that we only need to care about the action of $\Psi$ on two kinds of morphisms:
\begin{itemize}
    \item the \emph{restrictions} $\Sigma_{t, \mathcal{X}} \slicePreceq \Sigma_{t, \mathcal{Y}}$, where $\mathcal{Y} \subseteq \mathcal{X}$;
    \item the \emph{1-step evolutions} $\Sigma_{t, \mathcal{Y}} \slicePreceq \Sigma_{t+1, \mathcal{X}}$, where $\mathcal{Y} = \bigcup_{\underline{x} \in \mathcal{X}} \suchthat{\left(t, \underline{x} + \underline{\delta}\right)}{\underline{\delta} \in \mathcal{N}}$.
\end{itemize}
The existence of the factorisation above can be proven by induction, observing that any morphism $\Sigma_{t, \mathcal{X}_0} \slicePreceq \Sigma_{t+1, \mathcal{X}_1}$ factors into the product:
\begin{equation}
    \left(\Sigma_{t, \mathcal{Y}_0} \slicePreceq \Sigma_{t+1, \mathcal{X}_1}\right)
    \otimes
    \left(\Sigma_{t, \mathcal{X}_0 \backslash \mathcal{Y}_0} \slicePreceq \emptyset\right)
\end{equation}
where $\mathcal{Y}_0$ is defined as before so that $\Sigma_{t, \mathcal{Y}_0}$ is exactly the set of immediate causal predecessors of the codomain $\Sigma_{t+1, \mathcal{X}_1}$.

On restrictions $\Sigma_{t, \mathcal{X}} \slicePreceq \Sigma_{t, \mathcal{Y}}$, where $\mathcal{Y} \subseteq \mathcal{X}$, the functor $\Psi$ is defined to act by marginalisation, discarding the field state over all those events in the larger slice $\Sigma_{t, \mathcal{X}}$ which don't belong to the smaller slice $\Sigma_{t, \mathcal{Y}}$:
\begin{equation}
    \label{equation:partitioned-cca-restriction}
    \Psi(\Sigma_{t, \mathcal{X}} \slicePreceq \Sigma_{t, \mathcal{Y}})
    :=
    \bigotimes_{\underline{x} \in \mathcal{X}} F_{\underline{x}}
    \hspace{3mm}\text{ where }\hspace{3mm}
    F_{\underline{x}} :=
    \begin{cases}
        \id{\mathcal{H}^{\otimes \mathcal{N}}} &\text{ if } \underline{x} \in \mathcal{Y}\\
        \trace{\mathcal{H}^{\otimes \mathcal{N}}} &\text{ if } \underline{x} \notin \mathcal{Y}
    \end{cases}
\end{equation}

On 1-step evolutions $\Sigma_{t, \mathcal{Y}} \slicePreceq \Sigma_{t+1, \mathcal{X}}$, where $\mathcal{Y} = \bigcup_{\underline{x} \in \mathcal{X}} \suchthat{\left(t, \underline{x} + \underline{\delta}\right)}{\underline{\delta} \in \mathcal{N}}$, the functor $\Psi$ is defined to act by a combination of evolution by $U$ and marginalisation.
The evolution component is simply an application of $U$ to the state at each event of $\mathcal{Y}$:
\begin{equation}
    \label{equation:partitioned-cca-evolution}
    U^{\otimes\mathcal{Y}}:
    \mathcal{H}^{\otimes(\mathcal{N} \times \mathcal{Y})}
    \rightarrow
    \mathcal{H}^{\otimes(\mathcal{N} \times \mathcal{Y})}
\end{equation}
The marginalisation component then needs to go from the codomain $\mathcal{H}^{\otimes(\mathcal{N} \times \mathcal{Y})}$ of the map above to the desired codomain $\mathcal{H}^{\otimes(\mathcal{N} \times \mathcal{X})}$.
To do this, we recall that the $\mathcal{H}$ factor of $\mathcal{H}^{\otimes(\mathcal{N} \times \mathcal{X})}$ corresponding to a given $\underline{\delta} \in \mathcal{N}$ and a given $\underline{x} \in \mathcal{X}$ is intended to encode the component of the state at $(t+1, \underline{x})$ coming from $(t, \underline{x} + \underline{\delta})$.
Analogously, the $\mathcal{H}$ factor of $\mathcal{H}^{\otimes(\mathcal{N} \times \mathcal{Y})}$ corresponding to a given $\underline{\delta} \in \mathcal{N}$ and a given $\underline{y} \in \mathcal{Y}$ is intended to encode the component of the evolved state going to $(t+1, \underline{y}-\underline{\delta})$.
Hence to go from $\mathcal{H}^{\otimes(\mathcal{N} \times \mathcal{Y})}$ to $\mathcal{H}^{\otimes(\mathcal{N} \times \mathcal{X})}$ we need to discard all factors in $\mathcal{H}^{\otimes(\mathcal{N} \times \mathcal{Y})}$ corresponding to components of the evolved state which are not going to some $(\underline{y} - \underline{\delta}) \in \mathcal{X}$:
\begin{equation}
    \label{equation:partitioned-cca-evolution-discarding}
    \left(
    \bigotimes_{(\underline{\delta}, \underline{y}) \in \mathcal{N} \times \mathcal{Y}}
    \hspace{-3mm}
    F_{\underline{\delta}, \underline{y}}
    \right)
    : \mathcal{H}^{\otimes(\mathcal{N} \times \mathcal{Y})} \rightarrow \mathcal{H}^{\otimes(\mathcal{N} \times \mathcal{X})}
    \hspace{3mm}\text{ where }\hspace{3mm}
    F_{\underline{\delta}, \underline{y}} :=
    \begin{cases}
        \id{\mathcal{H}} &\text{ if } (\underline{y} - \underline{\delta}) \in \mathcal{X}\\
        \trace{\mathcal{H}} &\text{ if } (\underline{y} - \underline{\delta}) \notin \mathcal{X}
    \end{cases}
\end{equation}
Putting the evolution and marginalisation components together we get the action of $\Psi$ on 1-step evolutions:
\begin{equation}
    \label{equation:partitioned-cca-evolution-full}
    \Psi\left(\Sigma_{t,  \mathcal{Y}} \slicePreceq \Sigma_{t+1, \mathcal{X}}\right)
    :=
    \left(
    \left(\bigotimes_{(\underline{\delta}, \underline{y}) \in \mathcal{N} \times \mathcal{Y}}
    \hspace{-3mm}
    F_{\underline{\delta}, \underline{y}}\right)
    \circ
    U^{\otimes\mathcal{Y}}
    \right)
    : \mathcal{H}^{\otimes(\mathcal{N} \times \mathcal{Y})}
    \rightarrow
    \mathcal{H}^{\otimes(\mathcal{N} \times \mathcal{X})}
\end{equation}
By construction, the above is a $G$-invariant causal field theory, completing the definition of our partitioned causal cellular automaton. If $U$ is an isomorphism, the same construction on $\mathcal{C}^{rev}$ using $U^{-1}$ provides a causal reversal for $\Psi$, showing that the partitioned causal cellular automata above is reversible under those circumstances. Finally, Figure~\ref{figure:partitioned-cca} below depicts an example of action on morphisms for a $(1+1)$-dimensional partitioned causal cellular automaton.

\begin{figure}[h]
    \begin{center}
        \includegraphics[width=0.95\textwidth]{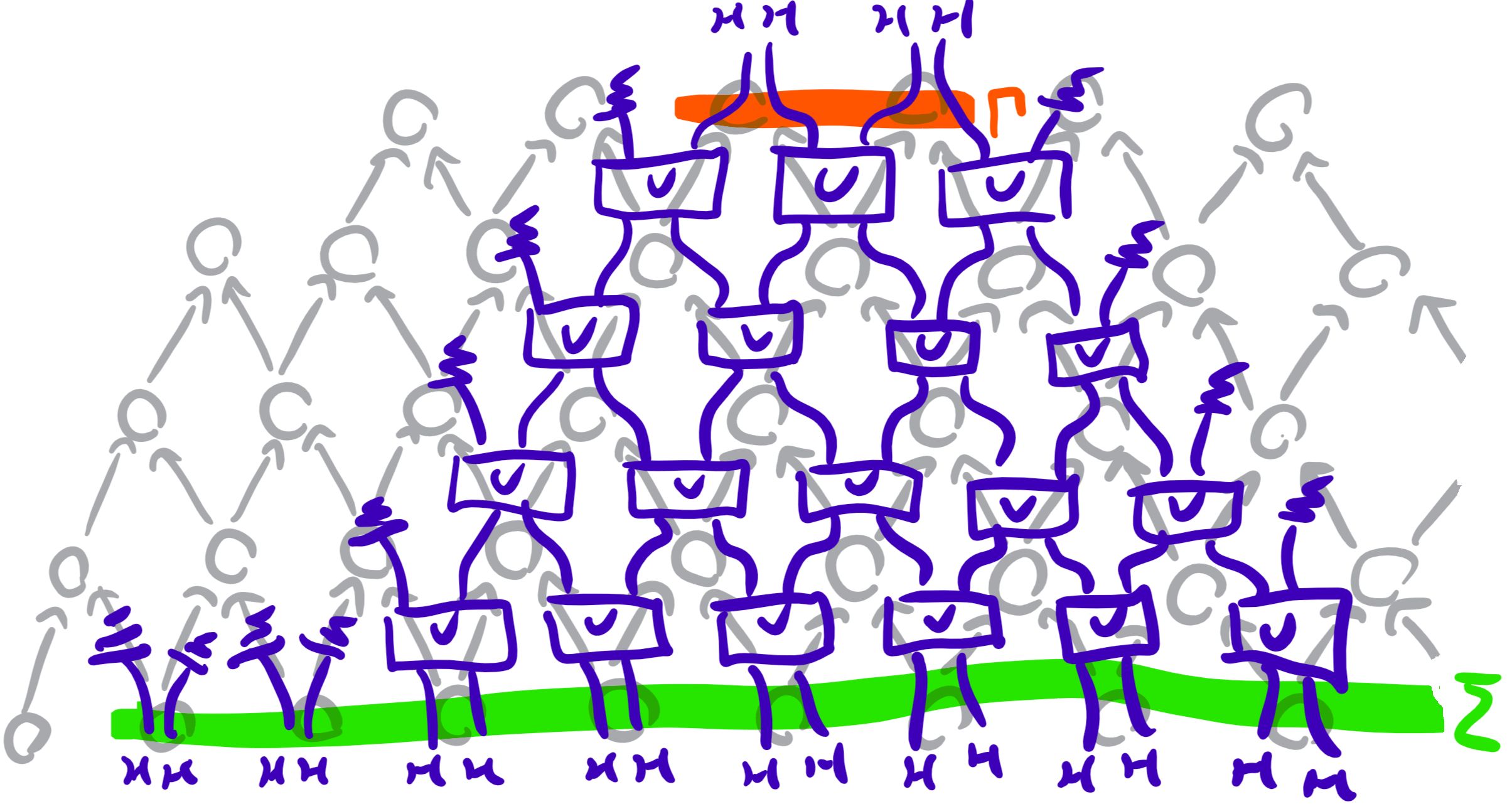}
    \end{center}
    \caption{
        Action of a partitioned causal cellular automaton over a complicated morphism $\Sigma \slicePreceq \Gamma$ in the $(1+1)$-dimensional example of the diamond lattice.
        Here $\mathcal{N} = \{\pm 1\}$, so each event in the causal order is associated to a copy of $\mathcal{H}^{\otimes \mathcal{N}} \cong \mathcal{H} \otimes \mathcal{H}$.
        The restriction action of the CCA (Equation~\ref{equation:partitioned-cca-restriction}) can be seen on the two events at the bottom left.
        The pure evolution action of the CCA (Equation~\ref{equation:partitioned-cca-evolution}) can be seen on the central pyramid of ten events, as the application of $U$ without discarding.
        The evolution + marginalisation action of the CCA (Equation~\ref{equation:partitioned-cca-evolution-full}) can be seen on the eight events at the sides of the central pyramid, as the application of $U$ followed by discarding of one of the two outputs.
        The input of the morphism depicted consists of eight copies of $\mathcal{H} \otimes \mathcal{H}$, one for each event of $\Sigma$, while the output of the morphism depicted consists of two copies of $\mathcal{H} \otimes \mathcal{H}$, one for each event of $\Gamma$.
    }
    \label{figure:partitioned-cca}
\end{figure}

\subsection{Sketch of the continuous limit for the Dirac QCA}

To conclude, we note how in \cite{arrighi2019automata} it is argued that the Dirac equation for free propagation of an electron can be recovered in the continuous limit of a specific $(1+1)$-dimensional partitioned QCA.
The original argument could not be made fully rigorous, because the QCAs defined therein were discrete and no setting was available to the author in which to make proper sense of the infinite tensor product arising from the limiting construction.
Our definition of CCA, on the other hand, has no requirement of discreteness.
Furthermore, the freedom left in the choice of field category for a CCA allows us to benefit from the full power of the non-standard approach to categorical quantum mechanics \cite{gogioso2017towards,gogioso2018quantum}.
As a consequence, we are able to sketch below a formalization in our framework of the continuous limit for the Dirac QCA, following the same lines as the construction of a $(1+d)$-dimensional partitioned CCA above.

The key to obtain a continuous limit for the Dirac QCA is to rescale the discrete lattice $\Omega$ to one with infinitesimal mesh $\varepsilon$:
\begin{equation}
    \varepsilon\Omega := \suchthat{\varepsilon(t, \underline{x})}{t \in \!\!\starIntegers, \underline{x} \in (t,...,t) + 2\starIntegers^{d}}
\end{equation}
where $\starIntegers$ are the non-standard integer numbers.
The slices are now allowed to contain an infinite number of points and can be used to approximate all equal-time partial Cauchy hyper-surfaces in $(1+d)$-dimensional Minkowski spacetime.
Unfortunately, the infinite number of points in our slices now requires infinite tensor products to be taken: to deal with this, we use as our field category the dagger compact category $\starHilbCat$ of non-standard hyperfinite-dimensional Hilbert spaces, where such infinite products can be handled safely.

We set the scattering map to be the following non-standard unitary
\begin{equation}
    U = 1 \oplus \sigma_X \exp(-im\varepsilon \sigma_X)\oplus 1
\end{equation}
where $\sigma_X$ is the $X$ Pauli matrix: this is the same unitary used in the Dirac QCA, but with the real parameter $\varepsilon$ turned into an infinitesimal.
Each application of $U$ only inches infinitesimally further from the identity, but in the non-standard setting we are allowed to consider the cumulative effect across infinite sequences of infinitesimally close slices.
The first order approximations to the Dirac equation derived in \cite{arrighi2019automata} turn into legitimate infinitesimal differentials, connecting the state on each slice to the state on the (infinitesimally close) next slice:
once the standard part is taken, the lattice $\varepsilon\Omega$ ends up covering the entirety of $(1+d)$-dimensional Minkowski spacetime, the differentials get integrated and $\Psi$ turns into a continuous-time field evolution following the Dirac equation.

\section{Conclusions and Future work}

In this work, we have defined a functorial, theory-independent notion of \emph{causal field theory} founded solely on the order-theoretic structure of causality.
We have seen how the causality requirement for such field theories is automatically satisfied as a consequence of symmetry-breaking in the ordering on space-like slices.
In an effort to connect to Algebraic Quantum Field Theory (AQFT), we have constructed complex spaces of states over regions of spacetime and discussed how the associated information redundancy can be reduced in selected cases.
We have introduced symmetries in our framework and shown that Quantum Cellular Automata (QCA) can be modelled within it, both in their traditional discrete formulation and in their continuous limit.

Despite our efforts, we feel we have barely scratched the surface on the potential of this material.
In the future, we envisage three lines of research stemming from this work. Firstly, we believe that the connection with AQFT can be strengthened and honed to the point that the framework will be a tool for the construction of new models.
This includes a thorough understanding of the structure of spaces of states for categories of slices more general than those induced by foliations.
Secondly, we wish to further explore and fully characterise the possibilities associated with working in the continuous limit of QCAs, with an eye to applications in perturbative quantum field theory.
Finally, we plan to extend the framework in a number of directions, including indefinite causal order, enrichment and the possibility of working with restricted classes of causal paths (in temporal analogy to categories of slices).

\section{Conflict of Interest Statement}

The authors declare that the research was conducted in the absence of any commercial or financial relationships that could be construed as a potential conflict of interest.

\section{Funding}

This work is supported by a grant form the John Templeton Foundation. The opinions expressed in this publication are those of the authors and do not necessarily reflect on the views of the John Templeton Foundation.
MS is supported by an EPSRC grant, ref. EP/P510270/1.

\bibliographystyle{plainurl}
\bibliography{bibliog}

\end{document}